\documentclass[numbers,authoryear, 12pt]{article}
\usepackage[T1]{fontenc}
\usepackage[utf8]{inputenc}
\usepackage[resetlabels]{multibib}
\newcites{oa}{References}
\usepackage[round]{natbib}
\usepackage{fancyhdr}
\usepackage{mathpazo}
\usepackage{setspace}
\usepackage[english]{babel}
\usepackage{a4wide}
\usepackage[pdftex]{graphicx}
\usepackage{pdfpages}
\usepackage[top=3cm, bottom=3cm, left=3cm, right=3cm]{geometry}
\usepackage{comment}
\usepackage{changepage}
\usepackage{graphicx}
\usepackage{paralist}
\usepackage{enumitem,kantlipsum}
\usepackage{listings}
\usepackage{nth}
\usepackage[english, ruled,vlined,resetcount]{algorithm2e}

\usepackage{algorithm2e}
\SetKwComment{Comment}{/* }{ */}
\SetAlgoCaptionSeparator{}
\usepackage{amsmath, amsthm, mathtools, amsfonts,mathrsfs, upgreek, mathabx, wasysym}
\numberwithin{equation}{section}
\usepackage{adjustbox}
\usepackage{setspace}
\usepackage{here}
\usepackage{rotating}
\usepackage{multirow}
\usepackage{booktabs}
\usepackage{lettrine}
\usepackage{dsfont}
\usepackage{textcomp}
\usepackage[tikz]{bclogo}
\usepackage{fancybox}
\usepackage{pifont}
\usepackage{color, xcolor}
\usepackage{bbm}
\usepackage{multicol}
\usepackage{footnote}
\usepackage{supertabular, longtable, caption, afterpage, tabularx, threeparttable}
\usepackage{array} % for defining custom column types
\newcolumntype{P}[1]{>{\centering\arraybackslash}p{#1}}
\counterwithin{figure}{section}
\counterwithin{table}{section}
\usepackage{colortbl}
\usepackage{dcolumn}
\newcolumntype{d}[1]{D..{#1}} % for alignment of numbers on decimal marker
\usepackage{siunitx}
\usepackage{lscape}

\usepackage[colorlinks=true,urlcolor= magenta,citecolor=blue]{hyperref}
\usepackage{titling}
\usepackage{xpatch}
\usepackage{titlesec}
\titleformat{\section}{\large\bfseries}{\thesection}{1em}{}
\titleformat{\paragraph}{\normalfont\normalsize\bfseries}{\theparagraph}{0.7em}{}
\titlespacing*{\paragraph}{0pt}{1ex plus 1ex minus .2ex}{1pt}
\usepackage{ragged2e}
\makeatletter
\newtheorem{theorem}{Theorem}[section]
\newtheorem{proposition}{Proposition}[section]
\newtheorem{lemma}{Lemma}[section]

\newcommand{\neutralize}[1]{\expandafter\let\csname c@#1\endcsname\count@}
\newtheorem{assumption}{Assumption}[section]

\newtheorem{example}{Example}[section]

\let\endtitlepage\relax
\newenvironment{mytitlepage}%
{\begin{titlepage}\def\@thanks{}}%
	{\end{titlepage}}
\xpatchcmd\titlepage{\setcounter{page}\@ne}{}{}{}
\xpatchcmd\endtitlepage{\setcounter{page}\@ne}{}{}{}
\newcommand{\ostar}{\mathbin{\mathpalette\make@circled\star}}
\newcommand{\make@circled}[2]{%
	\ooalign{$\m@th#1\smallbigcirc{#1}$\cr\hidewidth$\m@th#1#2$\hidewidth\cr}%
}
\newcommand{\smallbigcirc}[1]{%
	\vcenter{\hbox{\scalebox{0.77778}{$\m@th#1\bigcirc$}}}%
}
\DeclareMathOperator{\plim}{plim}

\definecolor{colari}{rgb}{0.7, 0, 0.7}
\definecolor{coland}{rgb}{0, 0.7, 0.4}

\makeatletter
\usepackage{authblk}
\newcommand{\TITLE}{Quantile Peer Effect Models}
\title{\vspace{-2cm}\TITLE\footnote[1]{\fontsize{10pt}{10pt} For comments and suggestions, I am grateful to Yann Bramoullé, Vincent Boucher, Firmin Doko Tchatoka, Mathieu Lambotte, and Marie Aurélie Lapierre. This research uses data from the National Longitudinal Study of Adolescent to Adult Health (Add Health), a program that is directed by Kathleen Mullan Harris and designed by J. Richard Udry, Peter S. Bearman, and Kathleen Mullan Harris at the University of North Carolina at Chapel Hill, and funded by Grant P01-HD31921 from the Eunice Kennedy Shriver National Institute of Child Health and Human Development, with cooperative funding from 23 other US federal agencies and foundations. Special acknowledgment is given to Ronald R. Rindfuss and Barbara Entwisle for assistance in the original design. Information on how to obtain Add Health data files is available on the Add Health website (\url{www.cpc.unc.edu/addhealth}). No direct support was received from Grant P01-HD31921 for this research.\\ 
		An R package, including all replication codes, is available at: \url{https://github.com/ahoundetoungan/QuantilePeer}.}}
\author{Aristide Houndetoungan}
\affil{\normalsize\emph{CY Cergy Paris Université}\\
	\href{mailto:aristide.houndetoungan@cyu.fr}{aristide.houndetoungan@cyu.fr}\vspace{-2pt}}
\date{\normalsize May 2025}
\begin{document}
	\setlength{\abovedisplayskip}{4pt}
	\setlength{\belowdisplayskip}{4pt}
	\begin{mytitlepage}
		\maketitle
		
		\vspace{-0.8cm}
		\begin{abstract}
			\noindent 
			{\linespread{1.1}\selectfont
				I propose a flexible structural model to estimate peer effects across various quantiles of the peer outcome distribution. The model allows peers with low, intermediate, and high outcomes to exert distinct influences, thereby capturing more nuanced patterns of peer effects than standard approaches that are based on aggregate measures. I establish the existence and uniqueness of the Nash equilibrium and demonstrate that the model parameters can be estimated using a straightforward instrumental variable strategy. Applying the model to a range of outcomes that are commonly studied in the literature, I uncover diverse and rich patterns of peer influences that challenge assumptions inherent in standard models. These findings carry important policy implications: key player status in a network depends not only on network structure, but also on the distribution of outcomes within the population.

				\vspace{0.5cm}
				
				%\vspace{0.2cm}
				%\noindent \textbf{JEL Classification}: 
			}
		\end{abstract}
	\end{mytitlepage}

	\newpage
	
	\section{Introduction}
	Following the seminal work of \cite{manski1993identification} on social interactions, numerous studies have examined how peers (or friends) influence individual behavior. The linear-in-means (LIM) model, which is commonly used in these analyses, is valued for its simplicity and well-established identification properties \citep{bramoulle2009identification, blume2015linear}. However, it captures peer influence solely through average peer outcome, thereby imposing the restrictive assumption that all peers exert the same effect, regardless of whether their outcomes are low or high. Such a strong restriction can lead to inaccurate counterfactual analyses and reduce the effectiveness of targeted policy interventions. %Recent contributions aimed at relaxing this assumption rely on specifications that still impose strong restrictions on the structure of peer effects or use fully parameterized (saturated) models in which each peer exerts a distinct effect. The latter often face identification challenges due to the large number of parameters involved.

	In this paper, I develop a structural model to estimate the effects of different quantiles of peer outcomes on individual behavior. The model allows peers with low, intermediate, and high outcomes to exert distinct effects. As is the case in \cite{boucher2024toward}, peer effects may arise from spillover effects, conformity effects, or a combination of the two. I establish the existence and uniqueness of a Nash equilibrium, despite the nonsmooth nature of the quantile function. I also show that the structural parameters can be estimated using a straightforward instrumental variable (IV) approach. Instruments for the quantiles of peer outcomes can be constructed from the quantiles of peers’ and peers-of-peers’ exogenous characteristics. I conduct an empirical analysis, showing that key players are likely to be misidentified by models that impose strong restrictions on the structure of peer effects. 
	
	%The model is flexible enough to capture heterogeneous effects across different points of the peer outcome distribution. 
	The sample quantile of peer outcomes serves as a local average around a given point in the peer outcome distribution. Measuring peer effects through these local averages at various quantile levels provides an efficient and straightforward way to characterize the relationship flexibly between an agent's outcome and those of their peers. Since outcomes are generally bounded in practice, a large number of quantile levels is not required for commonly observed network structures. This substantially reduces the number of parameters that must be estimated compared to fully parameterized (saturated) models, in which each peer exerts a distinct effect \citep{herstad2025identification}. Yet, selecting the relevant quantile levels is crucial to avoid misspecification. To address this, I develop an encompassing test that formally compares alternative specifications and helps identify the most appropriate set of quantile levels. As a practical tool for researchers, I also provide an easy-to-use R package, \texttt{QuantilePeer}, which includes routines that implement all of the methods developed in this paper.\footnote{The package is available at \href{https://github.com/ahoundetoungan/QuantilePeer}{https://github.com/ahoundetoungan/QuantilePeer}.}

	The model generalizes several specifications that measure peer effects at a single point in the peer outcome distribution, such as the minimum or maximum \citep{tao2014social, tatsi2015endogenous}. The standard LIM model also emerges as a special case, given that it implicitly assumes that all quantiles of peer outcomes exert the same effect. 
	
	Additionally, the proposed model offers greater flexibility than the specification of \cite{boucher2024toward}, which relies on a constant elasticity of substitution (CES) function. Because the CES function is always convex or concave, their specification can capture only monotonic peer effects; i.e., effects that increase, decrease, or remain constant as peer outcomes vary from low to high. However, it cannot accommodate non-monotonic influence patterns, in which mid-level outcome peers are the most influential, while those with extreme values have weaker effects, and vice versa, despite the potential prevalence of such patterns in many settings. Indeed, theories of social comparison \citep{festinger1954theory} suggest that aligning with moderately performing peers, rather than with outliers, can be socially or psychologically optimal. Such non-monotonic patterns can be captured by the quantile peer effect model.

	Despite the model’s flexibility, it remains linear in parameters, ensuring straightforward estimation. To illustrate this flexibility, I conduct a simulation study in which agents have up to ten friends. I show that the proposed model can accommodate rich peer effect structures that are not captured by the standard LIM or CES-based specifications. Moreover, in this simulation setting, saturated models would require a large number of coefficients, and the resulting estimates of peer effect parameters would exhibit high variance \citep[see][]{herstad2025identification}.

	I also present an empirical application using data from the National Longitudinal Study of Adolescent to Adult Health (Add Health). As is the case in the simulation study, individuals have up to ten friends. While estimating saturated models can be challenging, the encompassing test shows that three or four uniformly spaced quantile levels over the interval $[0, 1]$ are sufficient to capture peer effects flexibly. I examine several outcomes and find that a non-monotonic pattern of peer effects is observed for many outcomes, including academic achievement, participation in extracurricular activities, smoking, drinking, risky behaviors, and physical exercise. These findings align with social comparison theories and offer additional insights into the results that are obtained using the LIM and CES models.

	Using the empirical results, I identify key players as the nodes the removal of which from the network leads to the largest changes in the outcome distribution. Identifying such nodes is important, for example, for a social planner who is implementing targeted policies \citep[see][]{ballester2006s, lee2021key}. The results indicate that node rankings based on influence differ between the proposed model and the LIM and CES specifications, especially for outcomes that are characterized by non-monotonic peer effects, which these specifications fail to capture. Analyzing such outcomes using the LIM or CES models can distort the identification of key players, thereby reducing the effectiveness of targeted policy interventions.

	This paper contributes to the extensive literature on the identification and estimation of peer effects \citep{manski1993identification, lee2004asymptotic, bramoulle2009identification, blume2015linear, boucher2020estimating, lin2024quantile}. In particular, it engages with recent developments in semiparametric and nonlinear models of social interactions. One strand of this literature estimates peer effects based on latent traits or observable characteristics, such as race or gender \citep[e.g.,][]{comola2022heterogeneous, boucher2022peer, houndetoungan2024count}. Another approach models peer influence through nonlinear functions of peer outcomes, including the maximum, minimum, or median \citep{tao2014social, tatsi2015endogenous, diaz2021leaders}. Saturated models instead assign a distinct effect to each peer, but they often become impractical in dense networks due to the large number of parameters that are involved \citep{peng2019heterogeneous}. I contribute to this literature by generalizing aggregated social norm-based specifications. Despite its flexibility, the model remains linear in parameters and requires only a small number of coefficients, making it straightforward to estimate relative to saturated models.

	The two closest papers to this work are \cite{boucher2024toward} and \cite{herstad2025identification}. The former employs a CES function to link an agent’s outcome to those of their peers, while the latter estimates rank-dependent peer effects using a saturated model in which each peer exerts a distinct effect. The CES specification imposes a monotonic peer effect, while my model can capture richer patterns. The saturated model involves a large number of parameters, making estimation challenging in certain settings. In my empirical application, where agents have up to ten friends, the saturated model would require the estimation of 55 peer effect parameters, while my specification involves only three or four. Targeting quantiles of the peer outcome distribution, rather than modelling each peer individually, offers a more parsimonious representation while preserving flexibility. 
	
	%The encompassing test offers a practical tool to guide the selection of quantile levels. %The proposed model provides a practical alternative when peer effects are non-monotonic or when individuals belong to large peer groups.

	The remainder of the paper is organized as follows. Section \ref{sec:quantPE} introduces quantile peer effects models. Section \ref{sec:game} presents the microeconomic foundations of these models. Section \ref{sec:econometrics} addresses the identification and estimation of the model parameters, followed by a Monte Carlo study in Section \ref{sec:simulations}.
	Section \ref{sec:application} demonstrates the effectiveness of quantile peer effects models through an empirical application, and Section \ref{sec:conclusion} concludes the paper.

	\section{Quantile Peer Effects}\label{sec:quantPE}
	I consider a set $\mathcal{N}$ of $n$ agents that are indexed by $i \in [1, n]$. Agents are connected through a network that is represented by an adjacency matrix $\mathbf{G} = [g_{ij}]$ of dimension $n \times n$, where $g_{ij} = 1$ if agent $j$ is a friend (or peer) of agent $i$, and $g_{ij} = 0$ otherwise.\footnote{In weighted networks, $g_{ij}$ can take nonnegative values (not necessarily binary) to reflect the intensity of the outgoing link from $i$ to $j$. The results derived in this paper also apply to such networks.} The network may be directed; i.e., $g_{ij} = 1$ does not necessarily imply $g_{ji} = 1$. Self-links are not permitted: $g_{ii} = 0$ for all $i$. The outcome of agent $i$ is denoted by $y_i$.

	The linear-in-means (LIM) model is widely used to analyze how peers influence agent $i$'s outcome \citep[see][]{bramoulle2020peer}. In this model, $y_i$ is specified as:
	\begin{equation}\label{eq:LIM}
		y_i = c + \lambda \frac{\sum_{j = 1}^n g_{ij} y_j}{\sum_{j = 1}^n g_{ij}} + \boldsymbol x_i^{\prime}\boldsymbol\gamma + \eta_i,
	\end{equation}
	where $\boldsymbol{x}_i$ is a vector of exogenous characteristics, $\eta_i$ is an error term, and $c$, $\lambda$, and $\boldsymbol{\gamma}$ are unknown parameters. The term $\dfrac{\sum_{j = 1}^n g_{ij} y_j}{\sum_{j = 1}^n g_{ij}}$ represents the average outcome among peers. The parameter $\lambda$ captures endogenous peer effects, measuring the effect of a one-unit increase in the average peer outcome on the agent's outcome.
	
	I refer to Equation \eqref{eq:LIM} as the standard LIM model. This model relies on the restrictive assumption that all peers exert the same influence on agent $i$, regardless of whether the peers' outcome values are high or low. The influence of each peer $j$ on agent $i$ is given by $\frac{\partial y_i}{\partial y_j} = \lambda / \sum_{j = 1}^n g_{ij}$. To relax this assumption, I allow peer effects to vary across quantiles of peer outcomes.

	For any quantile level $\tau \in [0,1]$, I define the $\tau$-quantile outcome among peers as:
	\begin{equation} \label{eq:defqtaui}
		q_{\tau,i}(\mathbf y_{-i}) = Q_{\tau}\{y_j; ~ g_{ij} > 0\},
	\end{equation}
	where $y_{-i} = (y_1, y_2, \dots, y_{i - 1}, y_{i + 1}, \dots, y_n)^{\prime}$ is the vector of outcomes for individuals other than $i$, and $Q_{\tau}$ denotes the sample quantile at level $\tau$. Specifically, $q_{\tau,i}(\mathbf{y}_{-i})$ is the smallest value, such that at least a proportion $\tau$ of $i$'s peers have outcomes less than or equal to it. For instance, $q_{0,i}(\mathbf y_{-i})$ corresponds to the smallest peer outcome, $q_{1,i}(\mathbf y_{-i})$ to the largest peer outcome, and $q_{0.5,i}(\mathbf y_{-i})$ to the median peer outcome. For weighted networks, $Q_{\tau}$ can be a weighted quantile, where the weight that is assigned to $y_j$ is $g_{ij}$.\footnote{The literature commonly distinguishes nine types of sample quantiles \citep[see][]{hyndman1996sample}. The results that are derived in this paper hold for all of these types. Both the simulations and empirical analysis use \textit{Type 7}, which relies on linear interpolation when the quantile level does not correspond exactly to a peer's rank. For example, when an agent has only two friends, the sample median of peer outcomes is simply the average of the two values. The first decile is a weighted average of the two friends, with the friend who has the lower outcome receiving a weight of 0.9.} The sample quantile can be expressed as a weighted average of two peer outcomes, with the weights endogenously determined by the distribution of peer outcomes (see Appendix \ref{Append:sampleQ}). 
	
	Let $\mathcal{T} \subset [0, 1]$ be a finite set of quantile levels. I propose the following specification:
	\begin{equation}
		y_i = c + \sum_{\tau \in \mathcal{T}} \lambda_{\tau} q_{\tau,i}(\mathbf y_{-i}) + \boldsymbol x_i^{\prime}\boldsymbol\beta + \varepsilon_i, \label{eq:model:q}
	\end{equation}
	where $\lambda_{\tau}$ captures the effect of the $\tau$-quantile of peer outcomes on $y_i$, $\varepsilon_i$ is an error term, and $\boldsymbol\beta$ is a vector of parameters.  The quantile set $\mathcal{T}$ is chosen by the researcher. Importantly, misspecifying $\mathcal{T}$ can lead to incorrect counterfactual analyses. In Section~\ref{sec:econometrics:tests}, I propose an encompassing test to guide the selection of relevant quantile levels. It is important to select quantile levels uniformly over $[0, 1]$ to capture the effects of peers with both low and high outcomes.

	By including multiple values in $\mathcal{T}$, the model generalizes various specifications that measure peer effects at a single point in the distribution, such as the minimum and maximum \citep{tao2014social, tatsi2015endogenous}. The standard LIM model can also be seen as a special case, given that it implicitly assumes that $\lambda_{\tau}$ is constant.\footnote{Specifically, if $\lambda_{\tau} = \bar\lambda$ for all $\tau$, then $\sum_{\tau \in \mathcal{T}} \lambda_{\tau} q_{\tau,i}(\mathbf y_{-i}) = d_{t}\bar\lambda \left(\frac{1}{d_{t}}\sum_{\tau \in \mathcal{T}}q_{\tau,i}(\mathbf y_{-i})\right)$, where $d_{t}=\lvert \mathcal{T}\rvert$. If the quantile levels are chosen uniformly over $[0,1]$, so that $\frac{1}{d_{t}}\sum_{\tau \in \mathcal{T}}q_{\tau,i}(\mathbf y_{-i})$ approximates the average peer outcome, then $d_{t}\bar\lambda$ approximates the peer effect parameter in the LIM model.}

	Furthermore, the proposed model captures peer effects more flexibly than does the specification by \cite{boucher2024toward}, which replaces the average peer outcome in the LIM model with the following constant elasticity of substitution (CES) function:
	$$
	CES_i(\mathbf{y}_{-i}) = \left(\sum_{j=1}^{n} g_{ij} y_j^{\rho}\right)^{\frac{1}{\rho}}.
	$$
	The CES specification assigns greater weight to peers with high outcomes when $\rho > 1$, greater weight to peers with low outcomes when $\rho < 1$, and equal weight to all peers when $\rho = 1$, allowing peer effects to increase, decrease, or remain constant as peer outcomes vary from low to high.\footnote{This follows from the fact that $CES_i(\mathbf{y}_{-i})$ is strictly convex when $\rho > 1$, strictly concave when $\rho < 1$, and linear when $\rho =1$.} However, it does not account for potential non-monotonic influence patterns that may arise in certain settings. These include cases where peers with moderate outcomes exert greater influence than those with extreme values (high or low), and vice versa.

	\section{Microfoundations}\label{sec:game}
	In this section, I present the microfoundations of the model using a complete information game. Agent $i$ selects a strategy $y_i$ (e.g., academic effort) and derives utility from the distribution of their peers' strategies. Preferences are represented by a standard linear-quadratic utility function, which is commonly used in the context of the LIM model \citep[see][]{ballester2006s, calvo2009peer, de2017econometrics, bramoulle2020peer}, in which I replace the average peer effort with quantiles of the peer effort distribution.
	
	A well-known issue with the utility function, even in the LIM model, is that the social benefit term can be defined subjectively to capture either spillover effects or conformity \citep[see][]{boucher2016some}. I adopt the approach that was proposed by \cite{boucher2024toward}, which incorporates both spillover and conformity into the model, allowing their relative importance to be empirically identified from the data \citep[see also][]{lambotte2025peer}. The utility function is given by:
	\begin{equation}\label{eq:utility}
		U_i(y_i, ~ \mathbf y_{-i}) = \underbrace{\alpha_i y_i - \frac{1}{2}y_i^2}_{\text{private benefit}} + \underbrace{\sum_{\tau \in \mathcal T}\theta_{\tau,1} q_{i,\tau}(\mathbf y_{-i})y_i - \frac{1}{2}\sum_{\tau \in \mathcal T}\theta_{\tau,2}(y_i - q_{i,\tau}(\mathbf y_{-i}))^2}_{\text{social benefit}},
	\end{equation}
	where $\alpha_i$ is agent $i$'s characteristic (type) observable to all players, and $\theta_{\tau,1} \in \mathbb{R}$ and $\theta_{\tau,2} \geq 0$.\footnote{When $i$ has no friends (i.e., no outgoing links), the utility function excludes the social benefit term and is given by $U_i(y_i, \mathbf{y}_{-i}) = \alpha_i y_i - \frac{1}{2} y_i^2$. The corresponding optimal strategy is therefore $y_i = \alpha_i$. An individual without friends is referred to as isolated. However, in directed networks, such individuals may still be named as friends by others, i.e., they may have incoming links.} In the literature, the type $\alpha_i$ is generally specified as a linear function of observable characteristics $\boldsymbol{x}_i$ and an error term  \cite[e.g., see][]{blume2015linear}.

	The utility function is additively separable into a private benefit, which depends on agent $i$'s effort $y_i$, and a social benefit, which is a function of both $y_i$ and the effort vector of other agents, $\mathbf{y}_{-i}$. The marginal private utility is linear in the agent’s type $\alpha_i$, and the private cost (or disutility) of exerting effort $y_i$ is given by $\dfrac{1}{2} y_i^2$.

	The social benefit reflects both spillover effects (strategic complementarity or substitution) and conformity between an agent’s effort and the efforts of their peers. The term $\theta_{\tau,1} q_{i,\tau}(\mathbf y_{-i})y_i$ implies strategic complementarity between agent $i$’s effort and the $\tau$-quantile of peer efforts when $\theta_{\tau,1} > 0$. Conversely, if $\theta_{\tau,1} < 0$, then $\theta_{\tau,1} q_{i,\tau}(\mathbf y_{-i})y_i$ reflects strategic substitution (i.e., negative spillover) between efforts \citep{calvo2009peer, bramoulle2014strategic, houndetoungan2024identifying}. 
	The term $\theta_{\tau,2}(y_i - q_{i,\tau}(\mathbf y_{-i}))^2$ captures conformist preferences \citep{akerlof1997social, jackson2015games, ushchev2020social}. Given that $\theta_{\tau,2}$ is nonnegative, agent $i$ incurs a social cost when their effort deviates from the $\tau$-quantile of peer efforts.\footnote{The case where $\theta_{\tau,2}$ is negative (indicating anticonformity) is rare and is not considered in this paper, as it introduces complications in the analysis. Specifically, when $\theta_{\tau,2} < 0$, the utility function may become convex and lack a maximizer unless the strategy space is compact, in which case corner solutions may arise.} Preferences exhibit only conformity when $\theta_{\tau,1} = 0$ for all $\tau$, and only spillover effects when $\theta_{\tau,2} = 0$ for all $\tau$.
	
	For ease of exposition, I introduce the following notation:
	\begin{align*}
		&\lambda_{\tau,1} = \dfrac{\theta_{\tau,1}}{1 + \sum_{\tau \in \mathcal{T}} \theta_{\tau,2}}, \quad \lambda_{\tau,2} = \dfrac{\theta_{\tau,2}}{1 + \sum_{\tau \in \mathcal{T}} \theta_{\tau,2}}, \quad \text{and} \quad
		\lambda_{\tau} = \lambda_{\tau,1} + \lambda_{\tau,2} \text{ for all } \tau \in \mathcal{T}.\\
		&\lambda_{1} = \sum_{\tau \in \mathcal{T}}\lambda_{\tau,1}\quad \text{and} \quad \lambda_{2} = \sum_{\tau \in \mathcal{T}}\lambda_{\tau,2}.
	\end{align*}
	Agent $i$ chooses their effort $y_i$ by maximizing the utility function $U_i(y_i, \mathbf y_{-i})$. By solving the first-order condition of this maximization problem, I obtain the best response function:
	\begin{equation} \label{BRF}
		BR_i(\mathbf y_{-i}) = (1 - \lambda_{2})\alpha_i + \sum_{\tau \in \mathcal{T}} \lambda_{\tau} q_{i,\tau}(\mathbf y_{-i}).
	\end{equation}
	
	The parameter $\lambda_{\tau}$ measures the total peer effect at the $\tau$-quantile of the peer outcome distribution. This total effect is decomposed as $\lambda_{\tau} = \lambda_{\tau,1} + \lambda_{\tau,2}$, where $\lambda_{\tau,1}$ captures the spillover effect and $\lambda_{\tau,2}$ captures the conformity component. The total spillover effect across all quantiles is given by $\lambda_{1} = \sum_{\tau \in \mathcal{T}} \lambda_{\tau,1}$, while the total conformity parameter is $\lambda_{2} = \sum_{\tau \in \mathcal{T}} \lambda_{\tau,2}$. If $\mathbf y = (y_1, ~\dots, ~y_n)^{\prime}$ corresponds to a Nash equilibrium (NE), then $y_i = BR_i(\mathbf y_{-i})$ for all $i\in\mathcal{N}$. 
	
	The following proposition establishes the existence and uniqueness of the NE.
	\begin{proposition}\label{prop:NE} Assume that the utility function of each individual $i\in\mathcal{N}$ is given by \eqref{eq:utility}, with $\sum_{\tau \in \mathcal{T}}\lvert \lambda_{\tau} \rvert < 1$. Then, there exists a unique Nash equilibrium $\mathbf y^{\ast} = (y_1^{\ast}, ~\dots,   y_n^{\ast})^{\prime}$ such that $y_i^{\ast} = BR_i(\mathbf{y}^{\ast}_{-i})$.
	\end{proposition}
	\noindent One challenge in establishing the uniqueness is the non-differentiability of the quantile function. Nevertheless, I show that $BR(\mathbf y) = (BR_1(\mathbf y_{-1}), ~\dots, ~ BR_n(\mathbf y_{-n}))^{\prime}$ is a contracting mapping (see proof in Appendix \ref{Append:prop:NE}). This result is generalized to weighted networks in Online Appendix \ref{OA:WeightedNet}.

	The condition $\sum_{\tau \in \mathcal{T}} \lvert \lambda_{\tau} \rvert < 1$ ensures that a one-unit increase in peer outcomes does not lead to an increase in agent $i$'s effort greater than one. It rules out large quantile peer effects, which could cause the fixed point of the best response function $BR$ to diverge. When this condition is violated, a stable equilibrium in pure strategies does not generally exist. In the standard LIM model \eqref{eq:LIM}, a similar stability condition is $\lvert \lambda \rvert < 1$.
	
	As the best response function is a contraction, solving the game is straightforward. Under the condition $\sum_{\tau \in \mathcal{T}} \lvert \lambda_{\tau} \rvert < 1$, the Nash equilibrium $\mathbf{y}^\ast$ can be obtained by starting from any initial guess and iteratively applying the best response dynamics. The contraction property ensures convergence to the Nash equilibrium.

	\section{Identification and Estimation of the Model Parameters}\label{sec:econometrics}
	In peer effect models, interdependence between agents within a network can pose identification challenges. To address this issue, I follow the standard approach in the literature and assume that the network comprises multiple independent subnetworks (e.g., schools) that are not connected \citep{brock2007identification, bramoulle2009identification}.
	Let $S$ denote the number of subnetworks, $\mathcal{N}_s$ the set of agents in the $s$-th subnetwork, and $n_s \geq 2$ the number of agents in $\mathcal{N}_s$. I assume that $n_s$ is bounded. Henceforth, individual variables are indexed using a double subscript, $s$ and $i$, to indicate that they refer to agent $i\in\mathcal{N}_s$.

	%\subsection{Endogeneity and Instruments}
	\subsection{From a Theoretical Model to an Econometric Model}\label{sec:econometrics:model}
	In this section, I present the reduced-form equation for the outcome. For notational convenience, I drop the $\mathbf{y}_{s,-i}$ argument in $q_{s,i,\tau}(\mathbf{y}_{s,-i})$ and simply write $q_{s,i,\tau}$. Likewise, I omit the explicit notation $\tau \in \mathcal{T}$ when summing over quantile indices and write $\sum_{\tau}$. The agent type in the utility function \eqref{eq:utility} is specified as:
	\begin{equation}\label{eq:type}
		\alpha_{s,i} = c_s + \boldsymbol{x}_{s,i}^{\prime} \boldsymbol{\beta}_1 + \bar{\boldsymbol{x}}_{s,i}^{\prime} \boldsymbol{\beta}_2 + \varepsilon_{s,i},
	\end{equation}
	where $\boldsymbol{x}_{s,i}$ is a vector of observable characteristics of agent $i$ (e.g., social background), $\bar{\boldsymbol{x}}_{s,i}$ is a vector of exogenous peer characteristics (contextual variables), such as the average age and the share of girls within peers, and $\varepsilon_{s,i}$ is an error term \citep{blume2015linear}. The parameter vector $\boldsymbol{\beta}_1$ captures the influence of observable characteristics on agent types, such as how an increase in age affects a given type, whereas the parameter vector $\boldsymbol{\beta}_2$ captures the effects of contextual variables. Equation \eqref{eq:type} also accounts for unobserved subnetwork-level heterogeneity in agent type through fixed effects, which are represented by the term $c_s$.
	
	Even in the standard LIM model, the presence of \textit{isolated agents} in some subnetworks is necessary to disentangle spillover effects from conformity effects \citep{boucher2024toward}.\footnote{Agent $i$ is isolated if they have no friends, i.e., $g_{s,ij} = 0$ for all $j \in \mathcal{N}_s$. However, they may still be nominated by others (who are therefore non-isolated).} Let $\mathcal{N}_s^{iso}$ denote the set of isolated agents in subnetwork $s$, and $\mathcal{N}_s^{niso}$ the set of non-isolated agents. The optimal strategy of an isolated agent is given by $y_{s,i}^{iso} = \alpha_{s,i}$. Since the vector of contextual variables $\bar{\boldsymbol{x}}_{s,i}$ is not defined for isolated agents, their outcome can be expressed as:
	\begin{equation}\label{eq:model:iso}
		y_{s,i}^{iso} = c_s^{iso} + \boldsymbol{x}_{s,i}^{iso\prime} \boldsymbol{\beta}_1 + \varepsilon_{s,i}^{iso},
	\end{equation}
	where variables with the superscript $iso$ refer to isolated agents, and $c_s^{iso} = c_s$ denotes the subnetwork fixed effect for isolated agents. For non-isolated agents, the outcome follows the best response function~\eqref{BRF} and is given by:
	\begin{equation}\label{eq:model:niso}
		y_{s,i}^{niso} = c_s^{niso} + \sum_{\tau}\lambda_{\tau}q_{s,i,\tau}^{niso}  + \boldsymbol{x}_{s,iniso}^{\prime} \tilde{\boldsymbol{\beta}}_1 + \bar{\boldsymbol{x}}_{s,i}^{niso\prime} \tilde{\boldsymbol{\beta}}_2+ \varepsilon_{s,i}^{niso},
	\end{equation}
	where variables with the superscript $niso$ refer to non-isolated agents and $q_{s,i,\tau}^{niso}$ denotes the $\tau$-quantile outcome among the peers of non-isolated agent $i$. The parameter  $\tilde{\boldsymbol{\beta}}_1 = (1 - \lambda_2)\boldsymbol{\beta}_1$ is a vector of parameters associated with individual characteristics, $\tilde{\boldsymbol{\beta}}_2 = (1 - \lambda_2)\boldsymbol{\beta}_2$ is a vector associated with contextual variables, and  $c_s^{niso} = (1 - \lambda_2)c_s$ is the subnetwork fixed effect for non-isolated agents.

	%Preferences exhibit only conformism if $\lambda_{\tau,1} = 0$ (or $\theta_{\tau,1} = 0$) for all $\tau$ and only strategic complementarity or substitution if $\lambda_{\tau,2} = 0$ (or $\theta_{\tau,2} = 0$) for all $\tau$. The aggregate parameter $\theta_{2} = \sum_{\tau} \theta_{\tau,2}$ captures total peer effects through conformist behavior. Since $\theta_{\tau,2}$ is nonnegative, testing for conformist preferences reduces to testing whether $\theta_{2}$ is strictly positive. This test determines whether agents conform to their peers at any quantile level.
	
	The model captures diverse structures of peer effects and, consequently, richer mechanisms through which an intervention aimed at increasing $c_s$ (e.g., teacher quality) influences outcomes. If $\lambda_{2}$ is strictly positive, an increase in $c_s$ leads, in the first step, to a smaller increase in $y_{s,i}^{niso}$, since $c_s^{niso} = (1 - \lambda_{2})c_s$. However, due to peer effects, an additional increase arises through peers' outcomes, potentially generating social multiplier effects. This additional increase depends on the peer effects at each quantile, the distribution of peer outcomes, and the network structure. This mechanism differs from that of the LIM model, where social multiplier effects depend only on a single peer effect parameter and the network structure.

	\subsection{Instruments}\label{sec:econometrics:inst}
	Since peer outcomes are endogenous, identifying the parameters in Equation~\eqref{eq:model:niso} requires instruments for the quantile peer outcomes $q_{s,i,\tau}$. I introduce two types of instruments---Type I and Type II. The exogeneity of the Type I instruments is guaranteed as long as the network and the covariates in $\boldsymbol{x}_{s,i}$ are exogenous. The Type II instruments are potentially stronger, but they may fail to be exogenous even when the network and covariates are. In Section~\ref{sec:econometrics:tests:validity}, I present tests to assess whether the Type II instruments are exogenous, relying on the assumption that the Type I instruments are valid. If the Type II instruments are found to be exogenous, reporting results that are based on them, or on the combination of both types, is preferable, given that the corresponding estimator is likely to be more precise.

	The Type I instruments for $q_{s,i,\tau}$ are defined as the quantiles of the characteristics in $\boldsymbol{x}$ and the contextual variables in $\bar{\boldsymbol{x}}$ among peers. Equation~\eqref{eq:model:niso} implies that variation in peer characteristics and contextual variables explains variation in peer outcomes and, thus, in the quantiles of peer outcomes. Accordingly, quantiles of peer characteristics and contextual variables can serve as instruments for $q_{s,i,\tau}$. The set of quantile levels that are used to construct these instruments need not match $\mathcal{T}$. I recommend using a finer subdivision than $\mathcal{T}$ to improve the relevance of the instruments.

	For example, assume that $q_{s,i,\tau}$ is the median among peers. Due to the nonlinearity of the quantile function, the median of $\boldsymbol{x}$ and $\bar{\boldsymbol{x}}$ among peers cannot strongly predict $q_{s,i,\tau}$. However, using all deciles of $\boldsymbol{x}$ and $\bar{\boldsymbol{x}}$ among peers more effectively captures variation in $q_{s,i,\tau}$, as these deciles provide a more comprehensive representation of the distribution of $\boldsymbol{x}$ and $\bar{\boldsymbol{x}}$. Note that the same deciles are used as instruments for $q_{s,i,\tau}$ at all quantile levels; the instruments are not specific to $\tau$.\footnote{To further improve instrument relevance, one could also consider the quantiles of peers’ quantiles of $\boldsymbol{x}$ and $\bar{\boldsymbol{x}}$ within their own peer groups as additional instruments. This is similar to using the characteristics of more distant peers (e.g., peers of peers of peers) in the LIM model.}

	Since $\boldsymbol{x}$ and $\bar{\boldsymbol{x}}$ likely include many variables, a large number of quantiles may not be necessary to adequately capture variation in $q_{s,i,\tau}$. In small samples, however, a large number of instruments can lead to biased IV estimates. Fortunately, this bias can be corrected when the number of instruments grows at a rate of $\sqrt{n}$ or even $n$, using methods such as the jackknife IV estimator or post-estimation bias correction \citep[see][]{wang2016bootstrap, mikusheva2022inference, houndetoungan2024inference}. %As Equation \eqref{eq:model:niso} is linear in parameters, these methods apply directly. 
	My R package \href{https://github.com/ahoundetoungan/QuantilePeer}{\texttt{QuantilePeer}} provides routines to implement some of these methods.

	The type II instruments are also based on peer characteristics and contextual variables, with the key difference being that they focus on the peers whose outcomes determine $q_{s,i,\tau}$. As a result, the instruments are specific to each $\tau$-quantile variable $q_{s,i,\tau}$. As pointed out earlier, the sample quantile among peers can be expressed as a weighted average of two peer outcomes. Specifically,
	\begin{equation}\label{eq:qsitau}
		q_{s,i,\tau} = (1 - \omega_{\tau,s,i})y_{s,j_1}  + \omega_{\tau,s,i}y_{s,j_2},
	\end{equation}
	where $j_1$ and $j_2$ are two peers of agent $i$, and $\omega_{\tau,s,i} \in [0,1]$ is a weight. This representation holds for all types of quantiles and for any level $\tau$ (see Appendix~\ref{Append:sampleQ}). The indices $j_1$ and $j_2$, as well as the weight $\omega_{\tau,s,i}$, are endogenous and depend on the ranking of peer outcomes. For instance, if $\tau = 0$, then $\omega_{\tau,s,i} = 0$, and $j_1$ is the peer with the smallest outcome.

	From Equations \eqref{eq:model:niso} and \eqref{eq:qsitau}, it follows that $\boldsymbol{z}_{\tau,s,i} = (1 - \omega_{\tau,s,i})\boldsymbol{x}_{s,j_1}  + \omega_{\tau,s,i}\boldsymbol{x}_{s,j_2}$ and $\bar{\boldsymbol{z}}_{\tau,s,i} = (1 - \omega_{\tau,s,i})\bar{\boldsymbol{x}}_{s,j_1}  + \omega_{\tau,s,i}\bar{\boldsymbol{x}}_{s,j_2}$ are strong predictors of $q_{s,i,\tau}$, provided that $\tilde{\boldsymbol{\beta}}_1$ or $\tilde{\boldsymbol{\beta}}_2$ is nonzero. However, although $\boldsymbol{x}$ and $\bar{\boldsymbol{x}}$ are assumed to be exogenous, $\boldsymbol{z}_{\tau,s,i}$ may be endogenous because $j_1$, $j_2$, and the weight $\omega_{\tau,s,i}$ are themselves endogenous. 
	
	The exogeneity of $\boldsymbol{z}_{\tau,s,i}$ depends on the network structure. Specifically, it hinges on whether variation in $y_{s,i}^{niso}$ can influence peer outcome ranking. If the ranking changes, then $j_1$, $j_2$, and $\omega_{\tau,s,i}$ are also likely to change. In Section \ref{sec:econometrics:tests}, I present tests to assess this exogeneity, while relying on the assumption that the Type I instruments are exogenous. %It is also possible to combine both types of instruments to enhance relevance. The exogeneity of the resulting instrument set can also be tested.

	\subsection{Identification}
	In this section, I study the identification of the model parameters. Let $d_{t}$ denote the number of quantile levels and let $\mathcal{T} = \{\tau_1, ~\dots, ~\tau_{d_{t}}\}$. I define $\boldsymbol{v}_{s,i}^{niso} = \big(q_{s,i,\tau_1}^{niso}, ~\dots, \break q_{s,i,\tau_{d_{t}}}^{niso}, \boldsymbol{x}_{s,i}^{niso\prime}, \bar{\boldsymbol{x}}_{s,i}^{niso\prime}\big)^{\prime}$ as the vector of explanatory variables in Equation~\eqref{eq:model:niso}, and let $\boldsymbol{z}_{s,i}^{niso}$ denote the vector of instruments. Naturally, $\boldsymbol{z}_{s,i}^{niso}$ includes the exogenous regressors $\boldsymbol{x}_{s,i}^{niso}$ and $\bar{\boldsymbol{x}}_{s,i}^{niso}$, as well as the Type I or Type II instruments (or a combination of both types) for $(q_{s,i,\tau_1}^{niso}, \dots, q_{s,i,\tau_{d_t}}^{niso})$.

	Let $\mathbf{G}_s = [g_{s,ij}]$ denote the network matrix of the $s$-th subnetwork, and let $d_1$ and $d_2$ be the dimensions of $\boldsymbol{x}_{s,i}$ and $\bar{\boldsymbol{x}}_{s,i}$, respectively. I impose the following conditions.
	\begin{assumption}[Identification] \label{ass:identification} \hfill
		\begin{enumerate}[label=\Alph*., ref=\Alph*, nosep]
			\item \label{ass:identification:noocollinear} Linear independence: The matrices $\mathbb{V}ar(\boldsymbol{x}_{s,i}^{iso})$ and $\mathbb{V}ar(\boldsymbol{z}_{s,i}^{niso})$ are positive definite.
			\item \label{ass:identification:exo} Exogeneity: $\mathbb{E}(\varepsilon_{s,i}^{iso} \mid \boldsymbol{x}_{s,j}^{iso}) = 0$ and $\mathbb{E}(\varepsilon_{s,i}^{niso} \mid \mathbf{G}_s, \boldsymbol{z}_{s,j}^{niso}) = 0$ for all $i$ and $j\in\mathcal{N}_s$.
			\item \label{ass:identification:relevance} Instrument relevance: The covariance matrix $\mathbb{C}ov (\boldsymbol{v}_{s,i}^{niso}, \boldsymbol{z}_{s,i}^{niso})$ is full rank; that is, its rank is $d_{\nu} = d_{t} + d_1 + d_2$ (the dimension of $\boldsymbol{v}_{s,i}^{niso}$).
		\end{enumerate}
	\end{assumption}
	
	The conditions in Assumption~\ref{ass:identification} are standard. Condition~\ref{ass:identification:noocollinear} requires that $\boldsymbol{x}_{s,i}^{iso}$ and $\boldsymbol{z}_{s,i}^{niso}$ exhibit sufficient variation to avoid perfect collinearity after the variables in Equations~\eqref{eq:model:iso} and \eqref{eq:model:niso} are demeaned to eliminate the fixed effects $c^{iso}$ and $c^{niso}$. Condition~\ref{ass:identification:exo} implies that agents’ characteristics, instruments, and the network are exogenous. The Type I instruments involve only quantiles of $\boldsymbol{x}$ and $\bar{\boldsymbol{x}}$ among peers. Therefore, their exogeneity is ensured if the agents' characteristics $\boldsymbol{x}$, contextual variables $\bar{\boldsymbol{x}}$, and the network are exogenous.\footnote{Many papers have recently studied network endogeneity in linear peer effects models \citep[see][]{hsieh2016social, johnsson2021estimation}. Some of these methods involve controlling for additional explanatory variables that are unobserved, but can be identified from a network formation model. These methods can be applied directly to Equation~\eqref{eq:model:niso} due to its linearity in parameters.} However, the Type II instruments may not satisfy Condition \ref{ass:identification:exo}.

	Condition~\ref{ass:identification:relevance} implies that $\boldsymbol{z}_{s,i}^{niso}$ is a strong instrument for $\boldsymbol{v}_{s,i}^{niso}$. It also requires the components of  $\boldsymbol{v}_{s,i}^{niso}$ to be linearly independent. Depending on the network structure, there is therefore a limit to the number of quantile levels that can be included in $\mathcal{T}$. Including too many quantile levels can induce collinearity among the quantile peer variables, reducing the rank of $\mathbb{C}ov(\boldsymbol{v}_{s,i}^{niso}, \boldsymbol{z}_{s,i}^{niso})$ below $d_{\nu}$. Instrument strength can be assessed in practice. However, when there are multiple endogenous variables, the standard F-test for a single endogenous regressor is not applicable. Instead, one can rely on the rank Wald test proposed by \citet{kleibergen2006generalized}---henceforth, the KP test---to assess whether the rank of $\mathbb{C}ov(\boldsymbol{v}_{s,i}^{niso}, \boldsymbol{z}_{s,i}^{niso})$ is less than $d_{\nu}$.

	Not all parameters of the structural model are identifiable. In particular, $\theta_{\tau,1}$ and $\theta_{\tau,2}$ cannot be separately identified at each quantile level. However, this does not raise an issue for conducting counterfactual analyses. Indeed, in Equation~\eqref{eq:model:niso}, the key parameters that determine the effects of an exogenous shock are: $\lambda_2$, which captures the total peer effect across all quantile levels due to conformity; the peer effect at each quantile level, measured by $\lambda_\tau$; and the coefficients of the exogenous variables. The identification result is summarized in the following proposition.

	\begin{proposition}\label{prop:ident}
		Under the conditions of Proposition \ref{prop:NE} and Assumption \ref{ass:identification}, if the share of subnetworks with at least two isolated students does not converge to zero asymptotically, then $\theta_{\tau} := \theta_{\tau,1} + \theta_{\tau,2}$ for all $\tau \in \mathcal{T}$, $\theta_1 := \sum_{\tau} \theta_{\tau,1}$, and $\theta_2 := \sum_{\tau} \theta_{\tau,2}$, as well as $\boldsymbol{\beta}_1$ and $\boldsymbol{\beta}_2$, are identified.
	\end{proposition}
	The identification of $(\lambda_{\tau_1}, ~\dots,~ \lambda_{\tau_{d_{t}}})$, $\tilde{\boldsymbol{\beta}}_1$, and $\tilde{\boldsymbol{\beta}}_2$ from the reduced-form equation~\eqref{eq:model:niso} is standard, as is the case in IV models with fixed effects. Similarly, $\boldsymbol{\beta}_1$ can be identified from the reduced-form equation~\eqref{eq:model:iso} for isolated agents. The proof of Proposition~\ref{prop:ident} is straightforward, given the identification of $(\lambda_{\tau_1}, ~\dots,~ \lambda_{\tau_{d_{t}}})$, $\boldsymbol{\beta}_1$, $\tilde{\boldsymbol{\beta}}_1$, and $\tilde{\boldsymbol{\beta}}_2$ from the reduced forms.% (see Appendix \ref{Append:prop:ident}). 

	As is the case in \cite{boucher2024toward}, the presence of isolated agents plays an important role. The parameters $\boldsymbol{\beta}_1$ and $\tilde{\boldsymbol{\beta}}_1$ are identified separately from the isolated and non-isolated equations, respectively, which makes it possible to identify $\lambda_2$ from the condition $\tilde{\boldsymbol{\beta}}_1 = (1 - \lambda_2)\boldsymbol{\beta}_1$. If there are no isolated agents, it is still possible to identify $\lambda_{\tau}$ for all $\tau$ in Equation~\eqref{eq:model:iso}. However, in that case, it is not possible to determine whether peer effects arise from spillover or from conformity.

	\subsection{Estimation}\label{sec:econometrics:estim}
	The reduced-form parameters in Equations~\eqref{eq:model:iso} and \eqref{eq:model:niso} can be estimated using an IV approach combined with ordinary least squares (OLS), while controlling for subnetwork fixed effects. To eliminate the fixed effects $c_s^{iso}$ and $c_s^{niso}$, all individual-level variables should be demeaned by subtracting the subnetwork-specific average. Since $c_s^{iso} = c_s$ and $c_s^{niso} = (1 - \lambda_2)c_s$ are not identical, the variables for each group---isolated and non-isolated agents---should be demeaned using the average within their respective group.
	
	A standard OLS method can be applied to the demeaned version of Equation~\eqref{eq:model:iso} to estimate $\boldsymbol{\beta}_1$. This estimate can then be substituted into the demeaned version of Equation~\eqref{eq:model:niso}, using the relationship $\tilde{\boldsymbol{\beta}}_1 = (1 - \lambda_2)\boldsymbol{\beta}_1$. The resulting equation can be estimated using an IV approach or the generalized method of moments (GMM) to recover $(\lambda_{\tau_1}, \dots, \lambda_{\tau_{d_{t}}})$, $\lambda_2$, and $\boldsymbol{\beta}_2$. This estimation strategy relies on the assumption that $\boldsymbol{\beta}_1$ is the same in both Equations~\eqref{eq:model:iso} and \eqref{eq:model:niso}. This assumption can be relaxed to allow only a subset of components of $\boldsymbol{\beta}_1$ to be common across the two equations.

	This estimation approach follows a standard multi-stage procedure. The resulting estimator is asymptotically normal, and its variance can be consistently estimated. I provide an estimator of the asymptotic variance that accounts for the sampling error of the OLS estimator (see Online Appendix~\ref{OA:other:limit}).

	It is also possible to estimate both equations jointly using their respective moment conditions. For example, \citet{boucher2024toward} maximize the sum of the two GMM objective functions. This approach can improve precision, particularly for the estimate of $\boldsymbol{\beta}_1$. Yet, it requires numerical optimization, whereas the method discussed above yields closed-form estimates. Moreover, the weighting matrix of each GMM substantially influences the peer effect estimates, and care should be taken when defining these weights. Since the first stage relies on OLS, the precision gains from joint estimation may be marginal when the number of isolated agents is large.

	\subsection{Specification Tests} \label{sec:econometrics:tests}
	Several diagnostic tests are necessary to assess the validity of the model specification. As in standard IV approaches, these include tests for weak instruments and overidentifying restrictions. In this section, I introduce additional diagnostic tools that are tailored to the present context, including an encompassing test for the choice of quantile levels and a test for the exogeneity of the Type II instruments.
	
	\subsubsection{Choice of quantile levels}\label{sec:econometrics:tests:encompassing}
	Misspecifying the quantile levels in $\mathcal{T}$ can lead to incorrect counterfactual analyses and misguided policy recommendations. Let $\mathcal{T}_a$ and $\mathcal{T}_b$ denote two sets of quantile levels, which may or may not be nested. I propose an encompassing test \citep{smith1992non} to assess whether the model specified with $\mathcal{T}_a$ can replicate the features of the model specified with $\mathcal{T}_b$. Specifically, the test evaluates whether applying the specification based on $\mathcal{T}_b$ to data simulated from a model using $\mathcal{T}_a$ yields estimates that differ from those obtained when applying the same specification to the observed data. If the estimates differ, this suggests that some features of the observed data relevant for the specification with $\mathcal{T}_b$ are not captured by the model with $\mathcal{T}_a$. In that case, $\mathcal{T}_a$ must be rejected in favor of $\mathcal{T}_b$.

	Let $\mathbf{y}^{niso}$ denote the outcome vector for non-isolated agents. From Equation~\eqref{eq:model:niso}, a representation of $\mathbf{y}^{niso}$ under the quantile level set $\mathcal{T}_a$ is
	\begin{equation}\label{eq:yniso:mat}
		\mathbf{y}^{niso} = \mathbf{V}^{niso}_a \boldsymbol{\psi}_a + \boldsymbol{\varepsilon}^{niso}_a,
	\end{equation}
	where $\boldsymbol{\psi}_a = (\lambda_{\tau_1,a},~\dots,~\lambda_{\tau_{d_t},a},~1 - \lambda_{2,a},~\tilde{\boldsymbol{\beta}}_{2,a})^{\prime}$ is the vector of coefficients and $\mathbf{V}_a^{niso}$ is the matrix of explanatory variables. Specifically, the row of $\mathbf{V}_a^{niso}$ corresponding to agent~$i \in \mathcal{N}_s^{niso}$ is given by 
	$(q_{s,i,\tau_1}^{niso},~\dots,~q_{s,i,\tau_{d_t}}^{niso},~\boldsymbol{x}_{s,i}^{niso\prime} \boldsymbol{\beta}_1,~\bar{\boldsymbol{x}}_{s,i}^{niso\prime})$, where $d_t = \lvert \mathcal{T}_a \rvert$ is the number of quantile levels in $\mathcal{T}_a$. The vector $\boldsymbol{\varepsilon}^{niso}_a$ contains the error terms $\varepsilon_{s,i}^{niso}$ that are associated with $\mathcal{T}_a$. All variables in Equation~\eqref{eq:yniso:mat} are expressed in deviations from the group mean to eliminate fixed effects. Let $\mathbf{Z}_a^{niso}$ denote the instrument matrix corresponding to $\mathbf{V}_a^{niso}$. The notation with subscript~$a$ extends naturally to the quantile level set $\mathcal{T}_b$, with the subscript~$a$ being replaced by~$b$.
	
	Let $\hat{\boldsymbol{\psi}}_a$ denote the generalized method of moments (GMM) estimator of $\boldsymbol{\psi}_a$:
	$$
	\hat{\boldsymbol{\psi}}_a = \hat{\mathbf H}_a^{-1} \hat{\mathbf{V}}^{niso\prime}_a \mathbf{Z}_a^{niso} \mathbf{W}_a^{niso} \mathbf{Z}_a^{niso\prime} \mathbf{y}^{niso},
	$$
	where $\hat{\mathbf{V}}^{niso}_a$ is the matrix that is obtained by replacing $\boldsymbol{\beta}_1$ in $\mathbf{V}^{niso}_a$ with its OLS estimator,  $\mathbf{W}_a^{niso}$ is the GMM weighting matrix and $\hat{\mathbf H}_a = \hat{\mathbf{V}}^{niso\prime}_a \mathbf{Z}_a^{niso} \mathbf{W}_a^{niso} \mathbf{Z}_a^{niso\prime} \hat{\mathbf{V}}^{niso}_a$. In the special case of IV estimation, $\mathbf{W}_a^{niso} = \left(\mathbf{Z}_a^{niso\prime} \mathbf{Z}_a^{niso}\right)^{-1}$. The representation in Equation~\eqref{eq:yniso:mat} may not correspond to the true data-generating process, as it depends on the arbitrary choice of the quantile level set $\mathcal{T}_a$ and, therefore, on the auxiliary parameter $\boldsymbol{\psi}_a$. To complete the representation, I assume that $\boldsymbol{\psi}_a$ corresponds to the probability limit of its IV or GMM estimator; that is, $\boldsymbol{\psi}_a = \plim \hat{\boldsymbol{\psi}}_a$, where $\plim$ denotes convergence in probability as $S$ tends to infinity. This definition ensures that $\hat{\boldsymbol{\psi}}_a$ is a consistent estimator of $\boldsymbol{\psi}_a$, even if $\mathcal{T}_a$ is misspecified.

	If $\mathcal{T}_a$ is correctly specified, this representation should replicate the features of any alternative representation based on a different quantile level set $\mathcal{T}_b$. A practical way to assess this is to examine whether the estimator $\hat{\boldsymbol{\psi}}_b$ differs substantially when computed using the observed data versus data generated under the model based on $\mathcal{T}_a$.

	Let $\boldsymbol{\delta}_{a,b} = \plim \hat{\boldsymbol{\psi}}_b - \plim_{\mathcal{T}_a} \hat{\boldsymbol{\psi}}_b$, where $\plim_{\mathcal{T}_a} \hat{\boldsymbol{\psi}}_b$ denotes the probability limit of $\hat{\boldsymbol{\psi}}_b$ under the assumption that the data are generated according to the model with quantile set $\mathcal{T}_a$. The null hypothesis of the encompassing test is $\boldsymbol{\delta}_{a,b} = \mathbf{0}$. A consistent estimator of $\boldsymbol{\delta}_{a,b}$ is provided in the following proposition.
	
	\begin{proposition}\label{prop:encompassing}
		Let $\hat{\boldsymbol{\varepsilon}}_{a}^{niso} = \mathbf{y}^{niso} - \hat{\mathbf{V}}_a^{niso} \hat{\boldsymbol{\psi}}_a$. Under regularity conditions (see Assumption~\ref{ass:encompassing} in Appendix~\ref{Append:encompassing}), a consistent estimator of $\boldsymbol\delta_{a,b}$ is given by
		$$
		\hat{\boldsymbol{\delta}}_{a,b} = \hat{\mathbf{H}}_b^{-1}\, \hat{\mathbf{V}}_b^{niso\prime} \mathbf{Z}_b^{niso} \mathbf{W}_b^{niso}\mathbf{Z}_b^{niso\prime} \hat{\boldsymbol{\varepsilon}}_{a}^{niso}.
		$$
		%where $\hat{\mathbf{H}}_b = \hat{\mathbf{V}}_b^{niso\prime} \mathbf{Z}_b^{niso} \mathbf{W}_b^{niso} \mathbf{Z}_b^{niso\prime} \hat{\mathbf{V}}_b^{niso}$.
	\end{proposition}
	
	The proof of Proposition~\ref{prop:encompassing} is provided in Appendix~\ref{Append:encompassing}. Testing whether $\boldsymbol{\delta}_{a,b}$ differs from zero amounts to assessing whether $\hat{\mathbf{V}}_b^{niso}$ explains variation in the residual vector under $\mathcal{T}_a$. If $\hat{\mathbf{V}}_b^{niso}$ has explanatory power, then $\mathcal{T}_a$ is rejected. However, one cannot perform a simple regression of $\hat{\boldsymbol{\varepsilon}}_a$ on $\hat{\mathbf{V}}_b^{niso}$ to assess this explanatory power, given that the variables in this regression depend on both the OLS estimator $\hat{\boldsymbol{\beta}}_1$ and the GMM estimator $\hat{\boldsymbol{\psi}}_a$. Therefore, the test must account for the sampling error that is introduced by these preliminary estimators. I derive in Online Appendix~\ref{OA:other:encompassing} that $\sqrt{S}\, (\hat{\boldsymbol{\delta}}_{a,b} - \boldsymbol{\delta}_{a,b})$ is asymptotically normally distributed with mean zero, and I present its asymptotic covariance matrix, which accounts for this sampling error.

	Given the asymptotic distribution of $\sqrt{S}\, (\hat{\boldsymbol{\delta}}_{a,b} - \boldsymbol{\delta}_{a,b})$, a robust F-test can be used to test the joint significance of $\boldsymbol{\delta}_{a,b}$. Yet, for greater robustness, one can also use the rank test proposed by \citet{kleibergen2006generalized}, i.e., the KP test. Since $\boldsymbol\delta_{a,b}$ is a vector, its rank is either zero or one. The null hypothesis is that the rank is zero, meaning that no linear combination of the columns of $\hat{\mathbf{V}}_b^{niso}$ explains the residuals. As is the case in weak instrument tests, the KP test can detect whether the variables in $\hat{\mathbf{V}}_b^{niso}$ have collective explanatory power, even if each component of $\boldsymbol\delta_{a,b}$ is individually insignificant.

	\subsubsection{Validity of the Type II Instruments}\label{sec:econometrics:tests:validity}
	As pointed out in Section~\ref{sec:econometrics:inst}, the exogeneity of the Type II instruments can be tested under the assumption that the Type I instruments are exogenous. Let $\mathbf{Z}_1^{niso}$ and $\mathbf{Z}_2^{niso}$ denote the two instrument matrices for non-isolated individuals, and let $\boldsymbol{\varepsilon}^{niso}$ be the vector of error terms for non-isolated individuals. I propose two tests: a Sargan-style test and a Wald-style test. Both rely on the assumption that $\mathbf{Z}_1^{niso}$ is exogenous, i.e., $\mathbb{E}(\mathbf{Z}_1^{niso\prime} \boldsymbol{\varepsilon}^{niso}) / S = 0$. The null hypothesis is that $\mathbb{E}(\mathbf{Z}_2^{niso\prime} \boldsymbol{\varepsilon}^{niso}) / S = 0$.

	Let $\mathbf{P}_1 = \mathbf{Z}_1^{niso} (\mathbf{Z}_1^{niso\prime} \mathbf{Z}_1^{niso})^{-1} \mathbf{Z}_1^{niso\prime}$ denote the projection matrix onto the columns of $\mathbf{Z}_1^{niso}$, and let $\mathbf{M}_1 = \mathbf{I}_{niso} - \mathbf{P}_1$, where $\mathbf{I}_{niso}$ is the identity matrix of dimension $n^{niso} \times n^{niso}$ and $n^{niso}$ is the number of non-isolated individuals. The Sargan-style test exploits the decomposition
	$$
	\frac{\mathbb{E}(\mathbf{Z}_2^{niso\prime} \boldsymbol{\varepsilon}^{niso})}{S} = \frac{\mathbb{E}(\mathbf{Z}_2^{niso\prime} \mathbf{P}_1 \boldsymbol{\varepsilon}^{niso})}{S} + \frac{\mathbb{E}(\mathbf{Z}_2^{niso\prime} \mathbf{M}_1 \boldsymbol{\varepsilon}^{niso})}{S},
	$$
	where the term $\frac{\mathbf{Z}_2^{niso\prime} \mathbf{P}_1 \boldsymbol{\varepsilon}^{niso}}{S}$ can be expressed as $\frac{\mathbf{Z}_2^{niso\prime} \mathbf{Z}_1^{niso}}{S} \left( \frac{\mathbf{Z}_1^{niso\prime} \mathbf{Z}_1^{niso}}{S} \right)^{-1} \frac{\mathbf{Z}_1^{niso\prime} \boldsymbol{\varepsilon}^{niso}}{S}$. If $\mathbf{Z}_1^{niso}$ is exogenous, then $\frac{\mathbf{Z}_1^{niso\prime} \boldsymbol{\varepsilon}^{niso}}{S}$ is asymptotically zero. Consequently, a consistent estimator of $\frac{\mathbb{E}(\mathbf{Z}_2^{niso\prime} \boldsymbol{\varepsilon}^{niso})}{S}$ that accounts for the exogeneity of $\mathbf{Z}_1^{niso}$ is
	$$
	\mathcal{E} = \frac{\mathbf{Z}_2^{niso\prime} \mathbf{M}_1 \hat{\boldsymbol{\varepsilon}}^{niso}}{S},
	$$
	where $\hat{\boldsymbol{\varepsilon}}^{niso}$ denotes the residuals obtained using $\mathbf{Z}_1^{niso}$ as the instrument matrix.
	
	The test statistic is given by
	$$
	\mathcal{U}_1 = S \mathcal{E}^{\prime} \left( \hat{\mathbb{V}}ar(\sqrt{S} \mathcal{E}) \right)^+ \mathcal{E},
	$$
	where $\hat{\mathbb{V}}ar(.)$ denotes a consistent estimator of the asymptotic variance, and $(\cdot)^+$ denotes the Moore–Penrose pseudoinverse, given that $\hat{\mathbb{V}}ar(\sqrt{S} \mathcal{E})$ may be singular when the two instrument sets share columns. Under the null, $\mathcal{U}_1$ asymptotically follows a chi-squared distribution with degrees of freedom equal to the rank of $\mathbf{M}_1 \mathbf{Z}_2^{niso}$. This test differs from the standard Sargan test for overidentification, as it explicitly removes from $\mathbf{Z}_2^{niso}$ the variation that is linearly explained by $\mathbf{Z}_1^{niso}$, thereby testing whether the additional information in $\mathbf{Z}_2^{niso}$ is exogenous.

	The Wald-style test compares the parameter estimates that are obtained using each instrument matrix, in the spirit of the specification test proposed by \cite{hausman1978specification}. Let $\hat{\boldsymbol{\psi}}_1$ and $\hat{\boldsymbol{\psi}}_2$ be the estimators of $\boldsymbol{\psi} = (\lambda_{\tau_1},~\dots,~\lambda_{\tau_{d_t}},~1 - \lambda_{2},~\tilde{\boldsymbol{\beta}}_{2})^{\prime}$ using $\mathbf{Z}_1^{niso}$ and either $\mathbf{Z}_2^{niso}$ or the combination of $\mathbf{Z}_1^{niso}$ and $\mathbf{Z}_2^{niso}$, respectively. Since $\mathbf{Z}_1^{niso}$ is assumed to be exogenous, under the null hypothesis that $\mathbf{Z}_2^{niso}$ is also exogenous, the difference $\hat{\boldsymbol{\psi}}_1 - \hat{\boldsymbol{\psi}}_2$ is asymptotically zero. The test statistic is therefore given by
	$$
	\mathcal{U}_2 = S(\hat{\boldsymbol{\psi}}_1 - \hat{\boldsymbol{\psi}}_2)^{\prime} \left\{ \hat{\mathbb{V}}ar\Big(\sqrt{S}(\hat{\boldsymbol{\psi}}_1 - \hat{\boldsymbol{\psi}}_2)\Big) \right\}^+ (\hat{\boldsymbol{\psi}}_1 - \hat{\boldsymbol{\psi}}_2).
	$$
	
	Under the null, $\mathcal{U}_2$ asymptotically follows a chi-squared distribution with degrees of freedom equal to the dimension of $\hat{\boldsymbol{\psi}}_1$. A key distinction from the classical Hausman test is that neither estimator is efficient under the null. As a result, the variance $\hat{\mathbb{V}}ar\big(\sqrt{S}(\hat{\boldsymbol{\psi}}_1 - \hat{\boldsymbol{\psi}}_2)\big)$ is not proportional to the difference between the variances of the two estimators, as is the case in the standard Hausman test. Instead, the variance must be estimated from the joint variance of $\hat{\boldsymbol{\beta}}_1$, $\hat{\boldsymbol{\psi}}_1$, and $\hat{\boldsymbol{\psi}}_2$ (see Online Appendix~\ref{OA:other:variance}).
	
	\section{Monte Carlo Simulations}\label{sec:simulations}
	This section presents a simulation study to assess the finite-sample performance of the model. The network consists of $S = 50$ subnetworks, each containing $n_s = 50$ individuals. The distribution of the number of friends is calibrated to match that observed in the Add Health data, which is used in the empirical application. Individuals can have up to 10 friends; 22\% have no friends, the average number of friends is 3.47, and approximately 64\% of individuals have four friends or fewer.

	I consider two exogenous variables, $\boldsymbol{x}_1$ and $\boldsymbol{x}_2$, along with their corresponding contextual variables, $\bar{\boldsymbol{x}}_1$ and $\bar{\boldsymbol{x}}_2$, which represent the average values among friends. The variable $\boldsymbol{x}_1$ is simulated from a standard normal distribution, while $\boldsymbol{x}_2$ is drawn from a Poisson distribution with parameter 2. The parameter values are inspired by the estimates from the empirical application. I set $\boldsymbol{\beta}_1 = (-0.5, 1)$ and $\boldsymbol{\beta}_2 = (-0.2, 0.6)$. The model also includes an intercept, which is fixed at 4. The error term is drawn from a normal distribution with a zero mean and a standard deviation of 0.7.

	I consider quantile-based specifications with four quantile levels. The total peer effect parameter is set to $\lambda = 0.55$, with $\lambda_1 = 0.35$ and $\lambda_2 = 0.2$ representing spillover and conformity effects, respectively. I explore several data-generating processes (DGPs) that differ in how the total peer effect is distributed across quantiles. Table~\ref{tab:simu:estimation} reports the values of $\lambda_{\tau_1}$, $\lambda_{\tau_2}$, $\lambda_{\tau_3}$, and $\lambda_{\tau_4}$ for each DGP. Peer effects are increasing in DGP A, decreasing in DGP B, concave in DGP C, and convex in DGPs D and E. I also examine the performance of the quantile model when the data are generated from a standard LIM model (DGP F), where the peer effects are entirely due to spillover (i.e., $\lambda_2 = 0$) and set to $\lambda = 0.55$.

	Table~\ref{tab:simu:estimation} summarizes the estimation results. Since the true number of quantile levels is unknown in practice, I estimate models with three, four, and five quantile levels and use the encompassing test to compare their performance. For brevity, the estimates from the models with three and five quantile levels are not reported. However, Table~\ref{tab:simu:tests} presents the results of the encompassing tests, along with the KP test for weak instruments and the test for the validity of $\mathbf{Z}_2$ (the Type II instrument matrix), under the assumption that $\mathbf{Z}_1$ (the Type I instrument matrix) is exogenous. All estimates account for subnetwork-level heterogeneity using fixed effects.

	Overall, the quantile model performs well when either $\mathbf{Z}_1$ alone or the combination of $\mathbf{Z}_1$ and $\mathbf{Z}_2$ is used as the instrument set. The estimates are slightly more precise in the latter case, owing to the improved strength of the instruments. As shown in Table~\ref{tab:simu:tests}, the inclusion of $\mathbf{Z}_2$ leads to a substantial increase in the KP test statistic for weak instruments. Moreover, the validity of $\mathbf{Z}_2$ is generally not rejected at the 5\% or 10\% significance levels. One exception, however, is DGP~F, where the exogeneity of $\mathbf{Z}_2$ is rejected more frequently.

	In DGP F, this potential endogeneity is also reflected in the quantile peer effect estimates that are reported in Table~\ref{tab:simu:estimation}. Since the data in this case are generated from a standard LIM model, the peer effect at each quantile is $0.55 / 4 = 0.138$. Yet, the estimates exhibit a slight downward bias at the extreme quantiles and an upward bias at the second and third quantiles. One way to mitigate this bias is to increase the number of quantile levels. For instance, the encompassing test for DGP~F rejects the specification with four quantile levels in favor of the one with five in 20.6\% of simulations at the 5\% significance level.% and in 31.0\% of simulations at the 10\% significance level.

	Furthermore, the encompassing test comparing the specification with three quantile levels to that with four is rejected in many simulations for DGPs~A, B, C, E, and F. In particular, for DGP~E, where peer effects are convex and asymmetric, the three-quantile specification is always rejected. 
	
	Note that when the three-quantile specification is not rejected, this suggests that the specification captures all relevant features of the model with four quantile levels. This often occurs when quantile peer effects are monotonic, as in DGPs~A and B. Another example is DGP~D, where the three-quantile specification is rejected in only 1.0\% of simulations. This result is expected, as only peers with extreme values matter in this DGP. Thus, even two quantiles capturing the minimum and maximum are sufficient. Overall, the encompassing test results are coherent and indicate that it is a useful tool for selecting the quantile levels.

	Table~\ref{tab:simu:estimation} also presents estimation results from the LIM model and the CES-based specification. The $\rho$ parameter in the CES specification is relatively high (13.036) for DGP~A and low ($-4.119$) for DGP B. This aligns with the interpretation that peers with higher outcomes matter more in DGP A, whereas those with lower outcomes matter more in DGP B. The CES specification can serve as an alternative model for such DGPs, although it summarizes the peer effect pattern in a single parameter and masks the heterogeneity of effects across specific quantiles. In contrast, the LIM model generally fails to capture the total peer effect accurately, underestimating it in DGP~A and overestimating it in DGP~B.

	The CES model also fails to identify the correct social norm when peer effects are not monotonic. For example, the estimates of $\rho$ neither reflect that peers with extreme outcome values are not influential in DGP C, nor that mid-level outcome peers are not influential in DGP D. The problem is even more pronounced in DGP~E, where peer effects are convex. The CES model estimates $\rho$ at 2.336, suggesting that peer influence increases with peer outcomes. Yet, strong peer effects are observed only at the second quantile. This discrepancy stems from the asymmetry in peer effects and the negative effect estimated at the first quantile. Such differences in the estimates can lead to misidentification of key players.

	\afterpage{
		\begin{landscape}
			\begin{table}[htbp]
				\centering
				\footnotesize
				\caption{Monte Carlo Simulations --- Estimations}
				\label{tab:simu:estimation}
				\resizebox{21.5cm}{!}{
					\begin{threeparttable}
						\begin{tabular}{ccccclccccclcclccc}
							\toprule
							\multicolumn{5}{c}{Quantiles,   Instruments: $\mathbf{Z}_1$}                                           &            & \multicolumn{5}{c}{Quantiles,   Instruments: $\mathbf{Z}_1$, $\mathbf{Z}_2$}                           &            & \multicolumn{2}{c}{LIM}           &           & \multicolumn{3}{c}{CES}           \\
							$\lambda_{\tau_1}$ & $\lambda_{\tau_2}$ & $\lambda_{\tau_3}$ & $\lambda_{\tau_4}$ & \multicolumn{2}{l}{$\lambda_2$} & $\lambda_{\tau_1}$ & $\lambda_{\tau_2}$ & $\lambda_{\tau_3}$ & $\lambda_{\tau_4}$ & \multicolumn{2}{l}{$\lambda_2$} & $\lambda$   & \multicolumn{2}{l}{$\lambda_2$} & $\rho$  & $\lambda$ & $\lambda_2$ \\
							\midrule
							\multicolumn{18}{c}{DGP A: $\boldsymbol\lambda = (0, 0.05, 0.2, 0.3)$, $\lambda_2 = 0.2$}                                                                    \\[0.5ex]
							0.000   & 0.050   & 0.199   & 0.301   & 0.201   &  & -0.000  & 0.050   & 0.199   & 0.301   & 0.201   &  & 0.436   & 0.202   &  & 13.036  & 0.582   & 0.200   \\
							(0.007) & (0.016) & (0.031) & (0.021) & (0.018) &  & (0.006) & (0.013) & (0.022) & (0.015) & (0.018) &  & (0.018) & (0.021) &  & (1.989) & (0.014) & (0.019) \\[2.5ex]
							\multicolumn{18}{c}{DGP B: $\boldsymbol\lambda = (0.3, 0.2, 0.05, 0)$, $\lambda_2 = 0.2$}                                                                    \\[0.5ex]
							0.300   & 0.200   & 0.048   & 0.001   & 0.201   &  & 0.300   & 0.200   & 0.049   & 0.001   & 0.201   &  & 0.602   & 0.198   &  & -4.119  & 0.508   & 0.201   \\
							(0.018) & (0.047) & (0.058) & (0.027) & (0.018) &  & (0.012) & (0.028) & (0.035) & (0.017) & (0.018) &  & (0.021) & (0.020) &  & (1.795) & (0.030) & (0.019) \\[2.5ex]
							\multicolumn{18}{c}{DGP C: $\boldsymbol\lambda = (0, 0.275, 0.275, 0)$, $\lambda_2 = 0.2$}                                                                   \\[0.5ex]
							0.000   & 0.275   & 0.274   & 0.001   & 0.201   &  & -0.000  & 0.275   & 0.274   & 0.001   & 0.201   &  & 0.552   & 0.201   &  & 1.396   & 0.566   & 0.201   \\
							(0.010) & (0.025) & (0.042) & (0.025) & (0.018) &  & (0.008) & (0.018) & (0.027) & (0.016) & (0.018) &  & (0.012) & (0.019) &  & (0.131) & (0.013) & (0.019) \\[2.5ex]
							\multicolumn{18}{c}{DGP D: $\boldsymbol\lambda = (0.275, 0, 0, 0.275)$, $\lambda_2 = 0.2$}                                                                   \\[0.5ex]
							0.275   & -0.000  & -0.001  & 0.276   & 0.201   &  & 0.275   & 0.000   & -0.001  & 0.276   & 0.201   &  & 0.543   & 0.200   &  & 0.288   & 0.520   & 0.201   \\
							(0.012) & (0.029) & (0.046) & (0.026) & (0.018) &  & (0.009) & (0.021) & (0.030) & (0.017) & (0.018) &  & (0.013) & (0.019) &  & (0.207) & (0.015) & (0.019) \\[2.5ex]
							\multicolumn{18}{c}{DGP E: $\boldsymbol\lambda = (-0.05, 0.35, 0.15, 0.1)$, $\lambda_2 = 0.2$}                                                               \\[0.5ex]
							-0.050  & 0.350   & 0.149   & 0.101   & 0.201   &  & -0.050  & 0.350   & 0.149   & 0.100   & 0.201   &  & 0.540   & 0.202   &  & 2.336   & 0.590   & 0.200   \\
							(0.009) & (0.022) & (0.039) & (0.024) & (0.018) &  & (0.007) & (0.017) & (0.025) & (0.016) & (0.018) &  & (0.012) & (0.019) &  & (0.176) & (0.013) & (0.019) \\[2.5ex]
							\multicolumn{18}{c}{DGP F (LIM model): $\lambda = 0.55$, $\lambda_2 = 0$}                                                                                    \\[0.5ex]
							0.114   & 0.168   & 0.143   & 0.125   & 0.001   &  & 0.114   & 0.164   & 0.152   & 0.120   & 0.001   &  & 0.550   & 0.001   &  & 1.001   & 0.550   & 0.001   \\
							(0.009) & (0.020) & (0.036) & (0.023) & (0.023) &  & (0.007) & (0.016) & (0.026) & (0.016) & (0.023) &  & (0.010) & (0.023) &  & (0.093) & (0.011) & (0.023) \\ \bottomrule
						\end{tabular}
						\begin{tablenotes} % Start tablenotes
							\item[-] The models are simulated and estimated 1{,}000 times. Values without parentheses represent average peer effect estimates, while those in parentheses correspond to standard errors. The instrument matrix $\mathbf{Z}_1$ includes the quantiles of $\boldsymbol{x}$ and $\bar{\boldsymbol{x}}$ among friends, computed at ten levels uniformly spaced between 0 and 1. For both $\mathbf{Z}_1$ and $\mathbf{Z}_2$, I include the quantiles of friends at distances up to three (i.e., friends of friends of friends).
							\item[-] DGPs A--E are generated from the proposed quantile-based model, where $\mathcal{T} = \{0, ~1/3, ~ 2/3,~ 1\}$, and $\boldsymbol\lambda = (\lambda_{\tau_1}, ~\lambda_{\tau_2}, ~\lambda_{\tau_3}, ~\lambda_{\tau_4})$ is the vector of peer effects at each quantile. The total conformity parameter is given by $\lambda_2 = \frac{\sum_{\tau} \theta_{\tau,2}}{1 + \sum_{\tau} \theta_{\tau,2}}$ (see Equation \eqref{eq:model:niso}). DGP F follows the standard LIM model with only spillover effects, where $\lambda = 0.55$ and $\lambda_2 = 0$.  
							\item[-] In both the standard LIM and CES-based models, $\lambda$ measures the total peer effect, while $\lambda_2$ represents the conformity parameter. The parameter $\rho$ is the substitution parameter of the CES social norm.  
							\item[-] All estimations account for unobserved subnetwork heterogeneity using fixed effects.
						\end{tablenotes}
				\end{threeparttable}}
			\end{table}
		\end{landscape}

		\begin{table}[htbp]
			\centering
			\caption{Monte Carlo Simulations --- Tests}
			\label{tab:simu:tests}
			\footnotesize
			\begin{threeparttable}
				\begin{tabular}{P{1cm}P{1cm}P{1cm}P{1cm}P{.5cm}ccP{1cm}P{1cm}}
					\toprule
					\multicolumn{4}{c}{Encompassing tests} && \multicolumn{2}{c}{KP test statistic} & \multicolumn{2}{c}{Validity of $\mathbf{Z}_2$}\\[0.5ex]
					\multicolumn{2}{c}{3 qtls vs. 4 qtls} & \multicolumn{2}{c}{4 qtls vs. 5 qtls}  && $\mathbf{Z}_1$ & $\mathbf{Z}_1, ~\mathbf{Z}_2$   \\[0.5ex]
					5\% & 10\% & 5\% & 10\% &&&& 5\% & 10\%\\
					
					\midrule
					\multicolumn{9}{c}{DGP A}                                                            \\[0.5ex]
					0.601 & 0.721 & 0.004 & 0.015 &  & 725.275   & 2{,}812.540   & 0.031      & 0.067       \\
					&       &       &       &  & (310.523) & (1161.461) &            &             \\[1.5ex]
					\multicolumn{9}{c}{DGP B}                                                            \\[0.5ex]
					0.305 & 0.432 & 0.000 & 0.004 &  & 701.655   & 2349.305   & 0.041      & 0.090       \\
					&       &       &       &  & (506.386) & (1393.227) &            &             \\[1.5ex]
					\multicolumn{9}{c}{DGP C}                                                            \\[0.5ex]
					0.999 & 0.999 & 0.003 & 0.009 &  & 783.901   & 2969.431   & 0.031      & 0.077       \\
					&       &       &       &  & (389.414) & (1394.342) &            &             \\[1.5ex]
					\multicolumn{9}{c}{DGP D}                                                            \\[0.5ex]
					0.008 & 0.018 & 0.003 & 0.011 &  & 756.292   & 2655.145   & 0.035      & 0.082       \\
					&       &       &       &  & (428.106) & (1342.365) &            &             \\[1.5ex]
					\multicolumn{9}{c}{DGP E}                                                            \\[0.5ex]
					1.000 & 1.000 & 0.002 & 0.011 &  & 742.023   & 2698.948   & 0.035      & 0.073       \\
					&       &       &       &  & (354.167) & (1219.311) &            &             \\[1.5ex]
					\multicolumn{9}{c}{DGP F (LIM model)}                                                \\[0.5ex]
					0.987 & 0.994 & 0.206 & 0.310 &  & 853.174   & 2456.247   & 0.146      & 0.246       \\
					&       &       &       &  & (431.564) & (1101.169) &            &             \\\bottomrule
				\end{tabular}
				\begin{tablenotes}[para,flushleft]
					For the encompassing tests, the columns labeled ``$a$ qtls vs. $b$ qtls``, for integers $a$ and $b$, report the share of rejections of the null hypothesis that the model with $a$ quantile levels does not perform worse than the model with $b$ quantile levels, at the significance levels indicated in the third row. The ``KP test statistic`` columns report the average value of the KP test statistic, with the corresponding standard deviations shown in parentheses. The instrument matrix used is specified in the second row. The columns under ``Validity of $\mathbf{Z}_2$`` report the share of rejections of the null hypothesis that $\mathbf{Z}_2$ is exogenous, based on a Sargan-style test that assumes $\mathbf{Z}_1$ is valid. Results are shown for the significance levels listed in the third row.
				\end{tablenotes}
			\end{threeparttable}
	\end{table}}
	
	\vspace{0.5cm}
	\vspace{-0.5cm}
	\section{Empirical Application}\label{sec:application}
	In this section, I present an empirical application using the Wave I dataset from the National Longitudinal Study of Adolescent to Adult Health (Add Health). I examine several outcomes and show that the quantile model captures diverse patterns of peer effects that challenge existing specifications. To illustrate the effectiveness of the model, I also conduct a counterfactual analysis, demonstrating that ignoring these patterns can lead to the misidentification of key players and reduce the impact of targeted policies.

	\subsection{Add Health Data}
	Wave I of the Add Health survey provides nationally representative and detailed information on \nth{7}--\nth{12} graders from 144 schools during the 1994--1995 school year in the United States (US). Approximately 90{,}000 students completed a questionnaire covering demographics, family background, academic performance, health-related behaviors, and friendship links. Each respondent could nominate up to five male and five female best friends within the same school.

	%The network is restricted to the school level—students from different schools are not connected as friends.

	The dataset has some limitations that merit discussion. First, some referred friend identifiers are missing and are removed from the network, following standard practice in the literature using this dataset. As a result, students who appear isolated may not actually be isolated. I exclude such "false" isolates from the sample when none of their nominated friends are observed in the network. In addition, the observed degree is censored, as students can nominate up to ten friends. However, only 1.12\% of students reach this maximum. To maintain focus on the main objective of the paper, I do not address this censoring issue. Despite these limitations, it is worth noting that the Add Health dataset remains the most comprehensive network dataset that is currently available for studying peer effects.

	I study 11 outcomes, including grade point average (GPA), academic effort, participation in extracurricular activities, future expectations, trouble at school, smoking, drinking, risky behaviors, self-esteem, physical exercise, and fighting. All of these outcomes, except for future expectations, have been studied by \cite{boucher2024toward} using the CES model. The future expectations outcome is constructed as the sum of binary indicators for whether students believe they will live to age 35, avoid HIV/AIDS, graduate from college, and have a middle-class family income by age 30. I control for several exogenous variables, including student age, grade, sex, race, Hispanic ethnicity (Spanish-speaking), and mother's education and employment. I also control for contextual variables, defined as the average of the exogenous variables among a student’s friends. 
	
	After removing observations with missing values, the final sample comprises approximately 75{,}000 students from 141 schools. The average number of friends per student is 3.47; 22\% of students have no friends, and approximately 64\% have four friends or fewer.
	
	\subsection{Empirical Results}
	The estimation results for the quantile, LIM, and CES models are summarized in Table~\ref{tab:appresult}, and the specification tests are presented in Table~\ref{tab:appresult:test}.\footnote{Note that the estimates for the CES model differ from those reported by \cite{boucher2024toward} because their model does not control for contextual effects.} Given that isolated students play a key role in the model identification, I also estimate the reduced-form specification of the quantile model on the subsample of non-isolated students (see Table~\ref{tab:appresult:test}). The quantile peer effect estimates are similar to those from the structural model in Table~\ref{tab:appresult}, suggesting that using isolated students for the identification of the structural model does not distort the decomposition of peer effects across quantiles.
	
	I first examine the specification tests reported in Table~\ref{tab:appresult:test}. The $p$-value from the encompassing test indicates that the four-quantile specification is never rejected in favor of the five-quantile specification at the 5\% significance level. However, the model with three quantile levels is rejected for GPA, extracurricular activities, smoking, and risky behavior. To ensure a consistent specification across all outcomes, I adopt the four-quantile model for the entire analysis. The instrument matrix $\mathbf{Z}_1$ includes the quantiles of $\boldsymbol{x}$ and $\bar{\boldsymbol{x}}$ among friends, computed at ten levels uniformly spaced between 0 and 1.\footnote{Only $\mathbf{Z}_1$ is used in the empirical application, as the validity of $\mathbf{Z}_2$ is rejected for most outcomes.} The KP test for weak instruments yields large test statistics, suggesting that the model does not suffer from weak instrument problems. Finally, the overidentification test indicates that the instruments are generally valid, except for the academic effort outcome, for which the Sargan $p$-value is 0.032.
	
	The estimation results from the quantile peer effects model indicate that peer effects primarily stem from conformity for most outcomes. Exceptions include GPA, trouble at school, physical exercise, and fighting, for which spillover effects dominate conformity effects. In the cases of trouble at school and physical exercise, peer effects are driven solely by spillover, as the conformity effects are not statistically significant. These results are consistent with those obtained under the standard LIM and CES specifications.
	
	The decomposition of peer effects across peer outcome quantiles reveals several patterns. A prevalent pattern is one in which mid-level outcome peers are the most influential, while students appear nearly sensitive to peers with the lowest or highest outcomes. This pattern is observed for GPA, extracurricular activities, trouble at school, and physical exercise. For these outcomes, the LIM and CES specifications tend to overestimate peer effects, given that they cannot isolate the influence of mid-level outcome peers without also attributing weight to those at the extremes.
	
	For instance, in the case of extracurricular activities, the CES specification identifies a negative substitution parameter, with a standard error indicating a statistically significant difference from one. This result suggests that peers with lower outcomes are the most influential. Yet, the quantile model offers more nuanced insight: the most influential peers are those around the second quantile level, while peers with the lowest outcomes have little to no effect.

	Another related pattern is when, in addition to the influence of peers with intermediate outcomes, peers at one of the extreme quantiles also exert significant effects. This pattern is observed for smoking, drinking, and risky behaviors. Once again, the LIM and CES specifications cannot capture such nuanced structures. For these outcomes, the estimated substitution parameter indicates that the CES specification is not significantly different from the LIM specification.

	Furthermore, the results reveal cases in which peers with extreme outcomes are as influential as, or even more influential than, those with intermediate outcomes. This pattern is observed for future expectations, self-esteem, and fighting. For self-esteem, the CES specification yields a highly negative substitution parameter ($\rho = -14.041$). If this estimate were not accompanied by a substantial standard error, one might conclude that only peers with the lowest outcomes were influential. However, the quantile model shows that peers at the third quantile are also influential, while those at the second quantile are not. This result helps explain the large variance associated with the estimate for $\rho$. A similar pattern emerges with academic achievement, where $\rho$ is estimated at $-4.764$. The large variance associated with this estimate reflects the fact that, although peers with the lowest outcomes are the most influential, those with high outcomes also matter, as indicated by the quantile model.

	The differences in peer effect patterns across specifications highlight the importance of the quantile model. The fact that peers at specific quantiles may exert distinct effects has significant policy implications. The identification of key players in the network is sensitive to this feature and may vary across the three models.

	\begin{table}[htbp]
		\centering
		\caption{Empirical Results}
		\label{tab:appresult}
		\resizebox{\textwidth}{!}{
			\begin{threeparttable}
				\begin{tabular}{ccccccd{1}ccd{1}ccc}
					\toprule
					\multicolumn{6}{c}{Quantile} && \multicolumn{2}{c}{LIM} && \multicolumn{3}{c}{CES}\\
					$\lambda_{\tau_1}$ & $\lambda_{\tau_2}$ & $\lambda_{\tau_3}$ & $\lambda_{\tau_4}$ & $\lambda_1$ & $\lambda_2$ & & $\lambda_1$ & $\lambda_2$ & & $\rho$ & $\lambda_1$ & $\lambda_2$ \\
					\midrule
					\multicolumn{13}{c}{Academic   achievements (GPA)}                                                                 \\[0.5ex]
					0.089   & 0.119   & 0.644   & -0.102  & 0.471   & 0.279   &  & 0.542   & 0.283   &  & 0.509    & 0.530   & 0.283   \\
					(0.054) & (0.107) & (0.115) & (0.062) & (0.042) & (0.033) &  & (0.045) & (0.033) &  & (0.681)  & (0.046) & (0.033) \\[2.5ex]
					\multicolumn{13}{c}{Academic effort}                                                                               \\[0.5ex]
					0.142   & 0.147   & 0.153   & 0.092   & 0.217   & 0.316   &  & 0.269   & 0.319   &  & -4.764   & 0.141   & 0.314   \\
					(0.036) & (0.074) & (0.059) & (0.046) & (0.068) & (0.042) &  & (0.116) & (0.043) &  & (5.332)  & (0.120) & (0.043) \\[2.5ex]
					\multicolumn{13}{c}{Extracurricular   activities}                                                                  \\[0.5ex]
					-0.079  & 0.530   & 0.265   & -0.011  & 0.269   & 0.436   &  & 0.368   & 0.444   &  & -0.175   & 0.285   & 0.426   \\
					(0.086) & (0.132) & (0.084) & (0.021) & (0.060) & (0.051) &  & (0.065) & (0.051) &  & (0.426)  & (0.058) & (0.051) \\[2.5ex]
					\multicolumn{13}{c}{future expectations}                                                                             \\[0.5ex]
					0.152   & 0.132   & 0.179   & 0.070   & 0.118   & 0.414   &  & 0.257   & 0.430   &  & 0.978    & 0.232   & 0.427   \\
					(0.037) & (0.101) & (0.147) & (0.084) & (0.048) & (0.032) &  & (0.060) & (0.033) &  & (0.772)  & (0.064) & (0.032) \\[2.5ex]
					\multicolumn{13}{c}{Trouble at school}                                                                             \\[0.5ex]
					0.002   & 0.303   & 0.236   & 0.029   & 0.475   & 0.095   &  & 0.648   & 0.119   &  & 0.348    & 0.697   & 0.123   \\
					(0.074) & (0.118) & (0.078) & (0.044) & (0.090) & (0.074) &  & (0.110) & (0.074) &  & (0.412)  & (0.105) & (0.073) \\[2.5ex]
					\multicolumn{13}{c}{Smoking}                                                                                       \\[0.5ex]
					-0.102  & 0.368   & 0.374   & 0.116   & 0.145   & 0.611   &  & 0.243   & 0.600   &  & 1.320    & 0.158   & 0.594   \\
					(0.083) & (0.110) & (0.063) & (0.023) & (0.050) & (0.038) &  & (0.049) & (0.040) &  & (0.675)  & (0.129) & (0.040) \\[2.5ex]
					\multicolumn{13}{c}{Drinking}                                                                                      \\[0.5ex]
					0.113   & 0.075   & 0.235   & 0.081   & -0.016  & 0.520   &  & 0.131   & 0.521   &  & 0.245    & 0.190   & 0.519   \\
					(0.123) & (0.171) & (0.085) & (0.016) & (0.063) & (0.033) &  & (0.080) & (0.034) &  & (0.326)  & (0.101) & (0.034) \\[2.5ex]
					\multicolumn{13}{c}{Risky behaviors}                                                                               \\[0.5ex]
					-0.061  & 0.332   & 0.264   & 0.112   & 0.180   & 0.467   &  & 0.205   & 0.462   &  & 0.631    & 0.242   & 0.461   \\
					(0.101) & (0.165) & (0.097) & (0.024) & (0.049) & (0.031) &  & (0.062) & (0.031) &  & (0.370)  & (0.069) & (0.031) \\[2.5ex]
					\multicolumn{13}{c}{Self-esteem}                                                                                   \\[0.5ex]
					0.120   & 0.132   & 0.206   & -0.025  & 0.149   & 0.285   &  & 0.129   & 0.265   &  & -14.041  & 0.067   & 0.287   \\
					(0.055) & (0.111) & (0.082) & (0.026) & (0.087) & (0.071) &  & (0.165) & (0.075) &  & (20.394) & (0.105) & (0.071) \\[2.5ex]
					\multicolumn{13}{c}{Physical exercise}                                                                             \\[0.5ex]
					0.085   & 0.139   & 0.192   & -0.022  & 0.334   & 0.060   &  & 0.425   & 0.057   &  & 0.653    & 0.460   & 0.062   \\
					(0.047) & (0.077) & (0.082) & (0.047) & (0.062) & (0.034) &  & (0.115) & (0.035) &  & (0.471)  & (0.080) & (0.035) \\[2.5ex]
					\multicolumn{13}{c}{Fighting}                                                                                      \\[0.5ex]
					0.207   & -0.005  & 0.171   & 0.168   & 0.386   & 0.154   &  & 0.367   & 0.153   &  & 2.434    & 0.321   & 0.151   \\
					(0.082) & (0.126) & (0.077) & (0.031) & (0.062) & (0.043) &  & (0.080) & (0.043) &  & (1.131)  & (0.088) & (0.043) \\\bottomrule
				\end{tabular}
				\begin{tablenotes}[para,flushleft]
					Estimates are reported without parentheses, with standard errors shown in parentheses. The first row indicates the model used: quantile, LIM, or CES. The full table, including coefficients for the control variables, is available upon request.
				\end{tablenotes}
		\end{threeparttable}}
	\end{table}

	\begin{table}
		\centering
		\footnotesize
		\caption{Empirical Results --- Isolated Students and Specification Tests}
		\label{tab:appresult:test}
		\begin{threeparttable}
			\begin{tabular}{P{1.2cm}P{1.2cm}P{1.2cm}P{1.2cm}P{1.5cm}P{1.7cm}P{1.2cm}P{1.2cm}}
				\toprule
				\multicolumn{4}{c}{Isolated students only} & \multirow{2}{*}{\parbox{1.5cm}{\centering KP test \\ statistic}} & \multirow{2}{*}{\parbox{1.7cm}{\centering Sargan test \\ $p$-value}} & \multicolumn{2}{P{2.4cm}}{Encompassing test $p$-value} \\
				$\lambda_{\tau_1}$ & $\lambda_{\tau_2}$ & $\lambda_{\tau_3}$ & $\lambda_{\tau_4}$ & & & 3 vs. 4 & 4 vs. 5\\
				\midrule
				\multicolumn{8}{c}{Academic   achievements (GPA)}                              \\[0.25ex]
				0.073   & 0.165   & 0.636   & -0.107  & 3,614 & 0.204  & 0.002     & 0.370     \\
				(0.054) & (0.107) & (0.114) & (0.061) &       &        &           &           \\[1ex]
				\multicolumn{8}{c}{Academic effort}                                            \\[0.25ex]
				0.153   & 0.141   & 0.156   & 0.102   & 997   & 0.032  & 0.578     & 0.368     \\
				(0.036) & (0.073) & (0.060) & (0.046) &       &        &           &           \\[1ex]
				\multicolumn{8}{c}{Extracurricular   activities}                               \\[0.25ex]
				-0.064  & 0.519   & 0.269   & -0.005  & 2,945 & 0.055  & 0.005     & 0.768     \\
				(0.086) & (0.132) & (0.084) & (0.021) &       &        &           &           \\[1ex]
				\multicolumn{8}{c}{future expectations}                                          \\[0.25ex]
				0.143   & 0.134   & 0.195   & 0.055   & 2,075 & 0.223  & 0.958     & 0.983     \\
				(0.036) & (0.100) & (0.146) & (0.084) &       &        &           &           \\[1ex]
				\multicolumn{8}{c}{Trouble at school}                                          \\[0.25ex]
				0.019   & 0.275   & 0.249   & 0.035   & 1,283 & 0.341  & 0.984     & 0.079     \\
				(0.074) & (0.118) & (0.078) & (0.044) &       &        &           &           \\[1ex]
				\multicolumn{8}{c}{Smoking}                                                    \\[0.25ex]
				-0.095  & 0.357   & 0.372   & 0.125   & 1,461 & 1.000  & 0.000     & 0.144     \\
				(0.083) & (0.111) & (0.063) & (0.023) &       &        &           &           \\[1ex]
				\multicolumn{8}{c}{Drinking}                                                   \\[0.25ex]
				0.119   & 0.056   & 0.226   & 0.080   & 1,300 & 0.573  & 0.074     & 0.927     \\
				(0.122) & (0.170) & (0.084) & (0.016) &       &        &           &           \\[1ex]
				\multicolumn{8}{c}{Risky behaviors}                                            \\[0.25ex]
				-0.065  & 0.350   & 0.278   & 0.117   & 1,979 & 0.132  & 0.029     & 0.080     \\
				(0.101) & (0.165) & (0.096) & (0.024) &       &        &           &           \\[1ex]
				\multicolumn{8}{c}{Self-esteem}                                                \\[0.25ex]
				0.112   & 0.155   & 0.230   & -0.022  & 1,566 & 0.090  & 0.955     & 0.992     \\
				(0.055) & (0.110) & (0.081) & (0.026) &       &        &           &           \\[1ex]
				\multicolumn{8}{c}{Physical exercise}                                          \\[0.25ex]
				0.096   & 0.156   & 0.193   & -0.002  & 1,222 & 0.260  & 0.941     & 0.308     \\
				(0.047) & (0.077) & (0.082) & (0.047) &       &        &           &           \\[1ex]
				\multicolumn{8}{c}{Fighting}                                                   \\[0.25ex]
				0.228   & 0.001   & 0.202   & 0.182   & 2,184 & 0.591  & 0.450     & 0.750     \\
				(0.082) & (0.127) & (0.077) & (0.031) &       &        &           &          \\\bottomrule
			\end{tabular}
			\begin{tablenotes}[para,flushleft]
				Estimates are reported without parentheses; standard errors are shown in parentheses. The first four columns present results from the model estimated using only isolated students. The KP test refers to the rank Wald test for weak instruments, while the Sargan test is the Sargan overidentification test. For the encompassing test, the columns labeled ``$a$ vs. $b$'', for integers $a$ and $b$, test the null hypothesis that the specification with $a$ quantile levels does not perform worse than the specification with $b$ quantile levels.
			\end{tablenotes}
		\end{threeparttable}
	\end{table}
	
	\subsection{Measuring Student Influence}
	This section studies the influence of students within their school. Influence is measured by the change in the school's average outcome when the student's outgoing and incoming links are removed. This corresponds to a scenario in which the student is \textit{fully isolated}; that is, they have no friends and are not nominated by others. Students who are already fully isolated in the observed network have no influence. For students who are not fully isolated, removing their links can affect the outcome distribution at the game equilibrium by altering the peer sets of those who nominated them. This measure of influence is also considered by \citet{ballester2006s} and \citet{lee2021key}, who define the key player as the student with the greatest influence.
	
	I numerically compute the influence for each student by setting the model parameters to their estimated values. Within each school, I rank students by assigning the highest rank to the student with the largest influence. I then compare the rankings obtained from the quantile model to those from the LIM and CES models. Since the effect of removing a single student's links can be negligible in large networks, I focus on schools with fewer than 50 students.\footnote{In larger schools, a similar simulation exercise can be conducted by removing the links of a group of students rather than just one.}

	Let $\mathbf{G}_s$ be the $n_s \times n_s$ adjacency matrix of school $s$, and let $\mathbf{G}_s^{(i)}$ denote the matrix obtained by setting the $i$-th row and column of $\mathbf{G}_s$ to zero. Let $y_{j,s}^{(i)}$ be the outcome of student $j$ when the school network is $\mathbf{G}_s^{(i)}$. The influence of student $i$ is measured by
	$$
	P_{s,i} = \dfrac{1}{n_s} \sum_{j = 1}^{n_s} \left(y_{s,j} - y_{j,s}^{(i)}\right).
	$$
	
	The student ranks (normalized between 0 and 100) are presented in Figure~\ref{fig:counterfact} for selected outcomes. The rank gaps are substantial, especially for extracurricular activities, smoking, and drinking, where certain students who receive the highest influence scores under the CES and LIM models are assigned relatively low scores under the quantile model, and vice versa. These discrepancies reflect the non-monotonic pattern of peer effects associated with these outcomes. The quantile model assigns low ranks to students who are peers with extreme outcome values, as such peers exert little or no influence. In contrast, the LIM and CES specifications cannot isolate the influence of mid-level outcome peers without also attributing high weight to those at the extremes. As a result, students who are peers with extreme outcome values are not assigned low ranks under these models. %A similar issue arises for the other outcome, though it is less pronounced.

	For the other outcomes, the rank gaps are less pronounced because the pattern of peer effects is closer to a monotonic or uniform structure, which can be captured by the LIM and CES specifications.

	Overall, this counterfactual analysis reveals that key player status varies across the three models. In the standard LIM model, a high outcome for a student who is not fully isolated increases their influence \citep[see][]{ballester2006s}. However, in the quantile or CES models, a high outcome does not necessarily imply greater influence if the student is not an influential peer. The results also show that when peer effects follow a non-monotonic pattern, the CES model does not fully address the limitations of the LIM model and may misidentify the most influential students. The quantile model provides an alternative specification in such settings.

	\begin{figure}[!h]
		\centering
		\includegraphics[scale = 1]{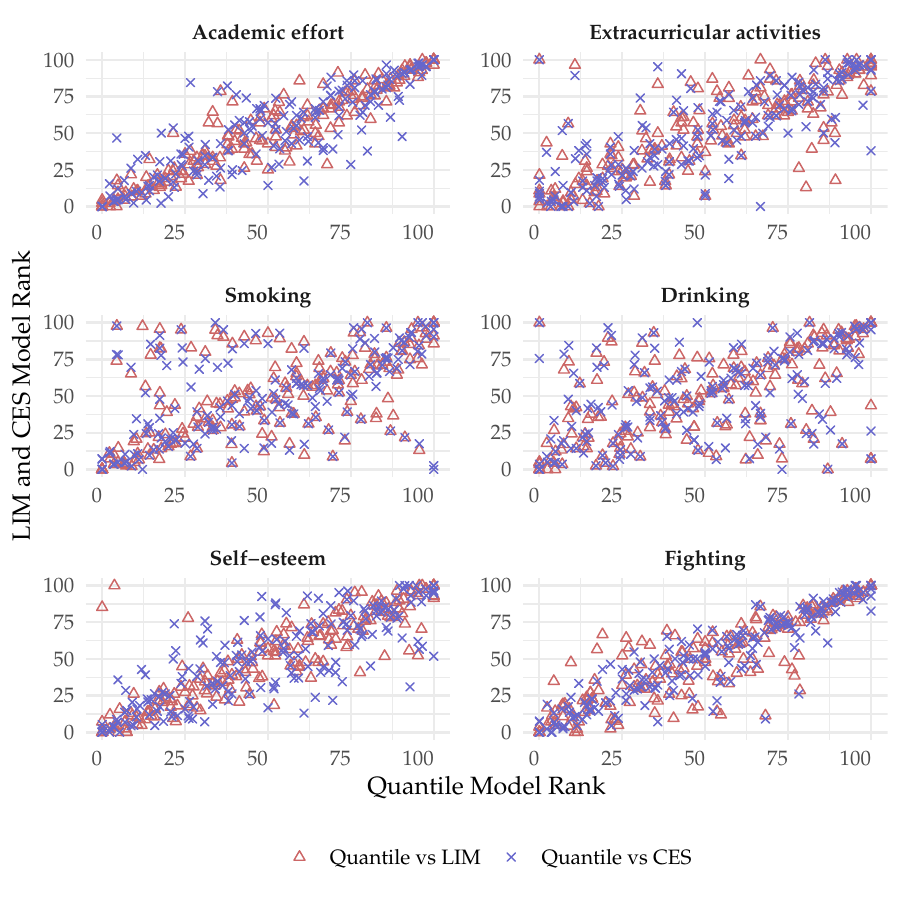}
		\caption{Influence Measure}
		\label{fig:counterfact}
		
		\vspace{-0.3cm}
		\footnotesize
		\justify
		The x-axis reports student ranks based on influence measure in the quantile model, while the y-axis reports ranks based on the LIM and CES models. Each red triangle represents a student’s LIM model rank (on the y-axis) against their quantile model rank (on the x-axis). Each blue “x” marker represents a student’s CES model rank (on the y-axis) versus their quantile model rank (on the x-axis).
	\end{figure}
	
	\section{Conclusion}\label{sec:conclusion}
	
	Social interactions are important because individuals’ decisions may be influenced by those of their peers or friends. Most papers that estimate peer effects rely on the standard LIM model, which assumes that each peer exerts the same influence on an individual. This assumption is restrictive, given that peers with high, intermediate, or low outcomes may have disproportionate influence depending on the context. To overcome this limitation, recent models, such as the CES-based specification by \citet{boucher2024toward}, allow peer effects to vary across the outcome distribution.

	This paper contributes to this literature by introducing a quantile peer effects model. Practitioners can define a set of quantile levels and estimate peer effects at each of these levels. The model is microfounded and captures peer influence arising from both spillover and conformity effects, as  is the case in the CES-based framework. A simple estimation strategy combining OLS and IV methods is proposed, along with an encompassing test to guide the selection of quantile levels. Monte Carlo simulations show that the test performs well and offers a practical tool for applied researchers.

	The proposed model is more flexible than the CES-based model, which restricts peer effects to vary monotonically across quantile levels. This flexibility is illustrated through an empirical application using multiple outcomes. The results reveal rich patterns of peer influence that are not adequately captured by either the CES or LIM specifications. In addition, the analysis of individual influence within the network shows that the restrictions imposed by these models can distort the identification of key players and limit the effectiveness of targeted policy interventions.
	
	Other models in the literature aim to capture heterogeneity in peer effects. For example, fully parameterized models allow each peer to exert a distinct influence. However, estimating such models can be challenging in practice due to the large number of parameters involved when agents have many friends. The quantile peer effects model offers a practical alternative. Since sample quantiles can be viewed as local averages of ranked peer outcomes, a large number of quantile levels is not required to capture peer influences adequately. Relying on these local averages of peer outcomes is conceptually similar to locally constant regressions in nonparametric models, where heterogeneity is captured through local variation.

	\newpage
	\appendix
	\renewcommand\thefigure{\thesection.\arabic{figure}}  
	\renewcommand\thetable{\thesection.\arabic{table}}  
	\setcounter{figure}{0} 
	\setcounter{table}{0} 
	\begin{center}
		\Large \bf Appendices
	\end{center}
	\section{Sample Quantiles}\label{Append:sampleQ}
	Let $d_i$ be the number of friends of agent $i$. For isolated agents (agents without peers), $d_i = 0$ and I set $q_{\tau,i}(\mathbf y_{-i}) = 0$. This decision is innocuous if the model includes separate intercepts for isolated and non-isolated agents.
	
	Assume that $d_i > 0$. I denote by $y_{i,(1)}$, \dots, $y_{i,(d_i)}$ the outcomes of agent $i$'s peers, ranked from the smallest to the largest, i.e., $y_{i,(1)} \leq \dots \leq y_{i,(d_i)}$. For all $\tau \in [0, 1]$, the sample quantile is defined as follows:
	\begin{equation}
		q_{\tau,i}(\mathbf y_{-i}) = (1 - \omega_i)y_{i,(\pi_i)}  + \omega_iy_{i,(\pi_i+1)}, \label{eq:defqtaui1}
	\end{equation}
	where $\pi_i = \lfloor\tau (d_i - 1)\rfloor + 1$ and $\omega_i = \tau (d_i - 1) - \lfloor\tau (d_i - 1)\rfloor$. The operator $\lfloor\cdot\rfloor$ denotes the floor function. Here, $\lfloor\tau (d_i - 1)\rfloor$ is the largest integer less than or equal to $\tau (d_i - 1)$. If $\tau (d_i - 1)$ is an integer, then $q_{\tau,i}(\mathbf y_{-i}) = y_{i,(\pi_i)}$; otherwise, $q_{\tau,i}(\mathbf y_{-i})$ is the weighted average of $y_{i,(\pi_i)}$ and $y_{i,(\pi_i + 1)}$. The closer $\tau (d_i - 1)$ is to $\lfloor\tau (d_i - 1)\rfloor$, the greater the weight of $y_{i,(\pi_i)}$. This definition of the sample quantile is popular and is known as the Type 7 sample quantile \citep[see][]{hyndman1996sample}. The case of weighted networks is presented in Online Appendix \ref{OA:WeightedNet}.
	
	%If $g_{ij}$ is not binary, then $y_{ij}$ can be weighted by $g_{ij}$, as is the case for average peer outcomes in the standard LIM model. Let $g_{i,(1)}$, \dots, $g_{i,(d_i)}$ the weights corresponding to $y_{i,(1)}$, \dots, $y_{i,(d_i)}$, respectively, and define $G_{i,l} = \sum_{k = 1}^{l} g_{i,(k)}$. Let $\hat{k}$ be the largest integer such that $\dfrac{G_{i,\hat k}}{G_{i,d_i}}  \leq \tau$. The weighted sample quantile is given by the same formula \eqref{eq:defqtaui1}, with the difference that: 
	%\begin{align*}
	%&\pi_i = \left\lfloor \hat k + 1 - \tau + (\tau G_{i,d_i} - G_{i,\hat k})/g_{i,(\hat k + 1)}\right\rfloor, \text{ and }\\
	%&\omega_i = \hat k + 1- \tau + (\tau G_{i,d_i} - G_{i,\hat k})/g_{i,(\hat k + 1)} - \pi_i.
	%\end{align*}
	%The higher the intensity of the outgoing link from $i$ to $\hat k + 1$, the higher the weight of $y_{i,(\hat k + 1)}$.

	\section{Existence and Uniqueness of Equilibrium}\label{Append:prop:NE}
	Let $\boldsymbol a = (a_1, ~\dots, ~ a_{\kappa})^{\prime} \in \mathbb R^{\kappa}$ be a $\kappa$-dimensional vector. For any integer $i \in [1, ~\kappa]$, let $\mathcal{R}_{\boldsymbol a}(a_i)$ denote the rank of the component $a_i$ within the vector $\boldsymbol a$. Specifically, $\mathcal{R}_{\boldsymbol a}(a_i) = 1$ if $a_i$ is the largest component of $\boldsymbol a$, $\mathcal{R}_{\boldsymbol a}(a_i) = 2$ if $a_i$ is the second largest, and so on, up to $\mathcal{R}_{\boldsymbol a}(a_i) = \kappa$ if $a_i$ is the smallest component of $\boldsymbol a$. In the case of ties, any tied component may be arbitrarily chosen as the larger rank, as this does not affect the quantiles of $\boldsymbol a$.
	
	The following lemma will be used in the remainder of this section.
	
	\begin{lemma}\label{lemmaRank}
		For any $\kappa$-dimensional  vectors $\boldsymbol a = (a_1, ~\dots, ~ a_{\kappa})^{\prime}$ and $\tilde{\boldsymbol a} = (\tilde a_1, ~\dots, ~ \tilde a_{\kappa})^{\prime}$, if  for some indices $i$ and $j$, $\mathcal{R}_{\boldsymbol a}(a_i) = \mathcal{R}_{\tilde{\boldsymbol a}}(\tilde a_j)$, there then exists an index $m$ such that $\lvert a_i - \tilde a_j\rvert \leq \lvert a_m - \tilde a_m \rvert$.
	\end{lemma}
	
	\begin{proof}
		I distinguish the following configurations: Case 1: $a_i \leq a_j$ and $\tilde a_j \leq \tilde a_i$; Case 2: $a_i \leq a_j$ and $\tilde a_i \leq \tilde a_j$; Case 3: $a_j \leq a_i$ and $\tilde a_i \leq \tilde a_j$; and Case 4: $a_j \leq a_i$ and $\tilde a_j \leq \tilde a_i$. I can ignore Cases 3 and 4, as they reduce to Cases 1 and 2, respectively, when $i$ and $j$ are switched.
		
		\medskip
		\noindent \underline{Case 1}: $a_i \leq a_j$ and $\tilde a_j \leq \tilde a_i$. This case implies that $a_i - \tilde a_j \leq a_j - \tilde a_j$ and $a_i - \tilde a_j \geq a_i - \tilde a_i$. Thus, $ a_i - \tilde a_i \leq a_i - \tilde a_j \leq a_i - \tilde a_j$ which means that $\lvert a_i - \tilde a_j \rvert \leq \max \{\lvert a_i - \tilde a_i \rvert, ~ \lvert a_j - \tilde a_j \rvert\}$. Consequently, there exists $m \in \{i, j\}$ such that $\lvert a_i - \tilde a_j\rvert \leq \lvert a_m - \tilde a_m \rvert$.

		\medskip
		\noindent \underline{Case 2}: $a_i \leq a_j$ and $\tilde a_i \leq \tilde a_j$. If $a_i = a_j$, then $a_i - \tilde a_j = a_j - \tilde a_j$. Thus, one can choose $m = j$. Similarly, if $\tilde a_i = \tilde a_j$, then $a_i - \tilde a_j = a_i - \tilde a_i$, and one can choose $m = i$. 
		
		Now, assume that $a_i \ne a_j$ and $\tilde a_i \ne \tilde a_j$; that is, $a_i < a_j$ and $\tilde a_i < \tilde a_j$. %It follows that $\boldsymbol a$ and $\tilde{\boldsymbol a}$ must have more than two components. Otherwise, $\boldsymbol a = (a_i, ~ a_j)^{\prime}$ and $\tilde{\boldsymbol a}  = (\tilde a_i, ~ \tilde a_j)^{\prime}$, then the conditions $a_i < a_j$ and $\tilde a_i < \tilde a_j$ would imply that $\mathcal{R}_{\boldsymbol a}(a_i) = 2$ and $\mathcal{R}_{\tilde{\boldsymbol a}}(a_j) = 1$, which would contradict the assumption that  $\mathcal{R}_{\boldsymbol a}(a_i) = \mathcal{R}_{\tilde{\boldsymbol a}}(a_j)$. Therefore, $\boldsymbol a$ and $\tilde{\boldsymbol a}$ must have at least three components. Specifically, since $a_i < a_j$ and $\tilde a_i < \tilde a_j$, 
		The only way for the condition $\mathcal{R}_{\boldsymbol a}(a_i) = \mathcal{R}_{\tilde{\boldsymbol a}}(a_j)$ to hold is that there exists an index $s$ such that:
		\begin{equation}\label{eq:us}
			a_s \leq a_i < a_j \quad \text{and} \quad 
			\tilde a_i < \tilde a_j \leq \tilde a_s.
		\end{equation}

		\noindent The existence of such an index $s$ is necessary because the condition $\mathcal{R}_{\boldsymbol a}(a_i) = \mathcal{R}_{\tilde{\boldsymbol a}}(a_j)$ involves the following two requirements: (1) the number of components in $\boldsymbol a$ that are greater than or equal to $a_i$ must be equal to the number of components in $\tilde{\boldsymbol a}$ that are greater than or equal to $\tilde a_j$; (2) the number of components in $\boldsymbol a$ that are less than or equal to $a_i$ must be equal to the number of components in $\tilde{\boldsymbol a}$ that are less than or equal to $\tilde a_j$. 
		It follows from \eqref{eq:us} that $a_s - \tilde a_s \leq a_i - \tilde a_j \leq a_j - \tilde a_j$. Thus, $\lvert a_i - \tilde a_j \rvert \leq \max \{\lvert a_j - \tilde a_j \rvert, ~ \lvert a_s - \tilde a_s \rvert\}$. Consequently, there exists $m \in \{j, s\}$ such that $\lvert a_i - \tilde a_j\rvert \leq \lvert a_m - \tilde a_m \rvert$. This completes the proof of the lemma.\end{proof}
	
	\paragraph{Proof of Proposition \ref{prop:NE}}
	The best response function is given by:
	$$BR_i(\mathbf y_{-i}) = (1 - \sum_{\tau \in \mathcal{T}} \lambda_{\tau,2})\alpha_i + \sum_{\tau \in \mathcal{T}}(\lambda_{\tau,1} + \lambda_{\tau,2})q_{i,\tau}(\mathbf y_{-i}).$$
	Let $BR$ be the mapping from $\mathbb R^n$ into itself, defined by $BR(\mathbf y) = (BR_1(\mathbf y_{-1}), ~\dots, \break BR_n(\mathbf y_{-n}))^{\prime}$. To establish the existence and uniqueness of the Nash equilibrium, it is sufficient to demonstrate that $BR$ has a unique fixed point. The challenge here is that $BR$ is not differentiable; therefore, I cannot show that its Jacobian has a norm less than one.
	
	Let $\mathbf y$  and $\tilde{\mathbf y}$ be two vectors of outcomes in $\mathbb R^n$. Let $\lVert . \rVert_{\infty}$ denote the infinity norm.\footnote{For any $\mathbf a = (a_1, ~\dots, ~ a_n)^{\prime} \in \mathbb R^n$, the infinity norm of $\mathbf a$ is defined as $\displaystyle \lVert \mathbf a \rVert_{\infty} = \max_{i} \lvert a_i \rvert$.}
	\begingroup
	\allowdisplaybreaks
	\begin{align}
		\lVert BR(\mathbf y) - BR(\tilde{\mathbf y}) \rVert_{\infty} &= \max_i \lvert BR_i(\mathbf y_{-i}) - BR_i(\tilde{\mathbf y}_{-i}) \rvert,\nonumber \\
		\lVert BR(\mathbf y) - BR(\tilde{\mathbf y}) \rVert_{\infty} & \textstyle \leq \sum_{\tau \in \mathcal{T}}\left\{\lvert \lambda_{\tau,1} + \lambda_{\tau,2}\rvert  \max_i  \lvert q_{i,\tau}(\mathbf y_{-i}) - q_{i,\tau}(\tilde{\mathbf y}_{-i}) \rvert\right\} \label{eq:BRy}
	\end{align}  
	\endgroup

	By Equation \eqref{eq:defqtaui1}, $q_{\tau,i}(\mathbf y_{-i}) = (1 - \omega_i)y_{i,(\pi_i)}  + \omega_iy_{i,(\pi_i+1)}$  and $q_{\tau,i}(\tilde{\mathbf y}_{-i}) = (1 - \omega_i)\tilde y_{i,(\pi_i)}  + \omega_i\tilde y_{i,(\pi_i+1)}$. Note that the weight $\omega_i$ and the index $\pi_i$ are the same for both $\mathbf y$  and $\tilde{\mathbf y}$ because $\omega_i$ and $\pi_i$ depend only on the network (particularly the degrees) and $\tau$.\footnote{The weight $\omega_i$ and the index $\pi_i$ do not depend on $\mathbf y$  and $\tilde{\mathbf y}$ because $q_{\tau,i}$ is an unweighted quantile. I generalize Proposition \ref{prop:NE} to the case of weighted networks in Online Appendix \ref{OA:WeightedNet}.} Thus,
	\begin{equation}\label{ineq:qtau}
		\lvert q_{i,\tau}(\mathbf y_{-i}) - q_{i,\tau}(\tilde{\mathbf y}_{-i}) \rvert \leq(1 - \omega_i)\lvert y_{i,(\pi_i)} -  \tilde y_{i,(\pi_i)} \rvert + \omega_i\lvert y_{i,(\pi_i+1)} - \tilde y_{i,(\pi_i+1)}\rvert.
	\end{equation}
	
	Since $\mathcal{R}_{\mathbf y}(y_{i,(\pi_i)}) = \mathcal{R}_{\tilde{\mathbf y}}(\tilde y_{i,(\pi_i)}) = \pi_i$, Lemma \ref{lemmaRank} implies that there exists an index $m$ such that $\lvert y_{i,(\pi_i)} -  \tilde y_{i,(\pi_i)} \rvert \leq \lvert y_{m} -  \tilde y_{m}\rvert$. Thus,  $\lvert y_{i,(\pi_i)} -  \tilde y_{i,(\pi_i)} \rvert \leq \lVert \mathbf y -  \tilde{\mathbf y}\rVert_{\infty}$. Similarly, since $\mathcal{R}_{\mathbf y}(y_{i,(\pi_i + 1)}) = \mathcal{R}_{\tilde{\mathbf y}}(\tilde y_{i,(\pi_i + 1)}) = \pi_i + 1$, there exists an index $p$, such that $\lvert y_{i,(\pi_i + 1)} -  \tilde y_{i,(\pi_i + 1)} \rvert \leq \lvert y_{p} -  \tilde y_{p}\rvert$, which implies  $\lvert y_{i,(\pi_i + 1)} -  \tilde y_{i,(\pi_i + 1)} \rvert \leq \lVert \mathbf y -  \tilde{\mathbf y}\rVert_{\infty}$. Therefore, \eqref{ineq:qtau} leads to:
	\begin{equation}\label{ineq:qtau2}
		\lvert q_{i,\tau}(\mathbf y_{-i}) - q_{i,\tau}(\tilde{\mathbf y}_{-i}) \rvert \leq\lVert \mathbf y -  \tilde{\mathbf y}\rVert_{\infty}.
	\end{equation}
	
	By substituting \eqref{ineq:qtau2} into \eqref{eq:BRy}, I obtain $\lVert BR(\mathbf y) - BR(\tilde{\mathbf y}) \rVert_{\infty}  \textstyle \leq  \sum_{\tau \in \mathcal{T}}\lvert \lambda_{\tau,1} + \lambda_{\tau,2}\rvert  \lVert \mathbf y -  \tilde{\mathbf y}\rVert_{\infty}$. Since $\sum_{\tau \in \mathcal{T}}\lvert \lambda_{\tau,1} + \lambda_{\tau,2}\rvert < 1$, this means that $BR$ is a contraction. By the contraction mapping theorem, it has a unique fixed point. As a result, the game has a unique Nash equilibrium.
	
	%\section{Identification of Model Parameters}\label{Append:prop:ident}
	%Given that $\tilde{\boldsymbol{\beta}}_1 = (1 - \lambda_2)\boldsymbol{\beta}_1$, the identification of both $\boldsymbol{\beta}_1$ and $\tilde{\boldsymbol{\beta}}_1$ implies that $\lambda_{2}$ is identified. Similarly, $\boldsymbol{\beta}_2$ is identified from the condition $\tilde{\boldsymbol{\beta}}_2 = (1 - \sum_{\tau} \lambda_{\tau,2})\boldsymbol{\beta}_2$, as both $\tilde{\boldsymbol{\beta}}_2$ and $\sum_{\tau} \lambda_{\tau,2}$ are identified.
	
	%Since $\lambda_{\tau,2} = \dfrac{\theta_{\tau,2}}{1 + \sum_{\tau} \theta_{\tau,2}}$, the identification of $\sum_{\tau} \lambda_{\tau,2}$ ensures that $\theta_2 = \sum_{\tau} \theta_{\tau,2}$ is identified. It also follows from $\lambda_{\tau,1} = \dfrac{\theta_{\tau,1}}{1 + \theta_2}$ and $\lambda_{\tau,2} = \dfrac{\theta_{\tau,2}}{1 + \theta_2}$ that $\theta_{\tau} = (1 + \theta_2)(\lambda_{\tau,1} + \lambda_{\tau,2})$. Consequently, $\theta_{\tau}$ is identified for all $\tau$, and $\theta_1 = \sum_{\tau} \theta_{\tau} - \theta_2$ is also identified.
	
	\section{Proof of Proposition \ref{prop:encompassing}}\label{Append:encompassing}
	The GMM estimator $\hat{\boldsymbol{\psi}}_b$ is given by
	\begin{equation}\label{eq:psihat}
		\hat{\boldsymbol{\psi}}_b = \hat{\mathbf{H}}_b^{-1} \hat{\mathbf{V}}_b^{niso\prime} \mathbf{Z}_b^{niso} \mathbf{W}_b^{niso} \mathbf{Z}_b^{niso\prime} \mathbf{y}^{niso},
	\end{equation}
	where $\hat{\mathbf{H}}_b = \hat{\mathbf{V}}_b^{niso\prime} \mathbf{Z}_b^{niso} \mathbf{W}_b^{niso} \mathbf{Z}_b^{niso\prime} \hat{\mathbf{V}}_b^{niso}$. To determine $\plim_{\mathcal{T}_a} \hat{\boldsymbol{\psi}}_b$, I assume that the data are generated according to the model with quantile set $\mathcal{T}_a$. In this case, the true data-generating process is $\mathbf{y}^{niso} = \mathbf{V}_a^{niso} \boldsymbol{\psi}_a + \boldsymbol{\varepsilon}_a^{niso}$, with $\boldsymbol{\varepsilon}_a^{niso} = \boldsymbol{\varepsilon}^{niso}$. Consequently, both $\mathbf{Z}_a^{niso}$ and $\mathbf{Z}_b^{niso}$ are exogenous with respect to $\boldsymbol{\varepsilon}_a^{niso}$. Substituting the expression for $\mathbf{y}^{niso}$ into Equation~\eqref{eq:psihat} yields:
	$$
	\hat{\boldsymbol{\psi}}_b = \hat{\mathbf{H}}_b^{-1} \hat{\mathbf{V}}_b^{niso\prime} \mathbf{Z}_b^{niso} \mathbf{W}_b^{niso} \mathbf{Z}_b^{niso\prime} \mathbf{V}_a^{niso} \boldsymbol{\psi}_a + \hat{\mathbf{H}}_b^{-1} \hat{\mathbf{V}}_b^{niso\prime} \mathbf{Z}_b^{niso} \mathbf{W}_b^{niso} \mathbf{Z}_b^{niso\prime} \boldsymbol{\varepsilon}_a^{niso}.
	$$

	The second term on the right-hand side is asymptotically zero because $\mathbf Z_b^{niso}$ is exogenous with respect to $\boldsymbol{\varepsilon}_a^{niso}$. Indeed, this term can be written as:
	$$
	\left(\frac{\hat{\mathbf{V}}_b^{niso\prime} \mathbf Z_b^{niso}}{S} \mathbf W_b^{niso} \frac{\mathbf Z_b^{niso\prime} \hat{\mathbf{V}}_b^{niso}}{S}\right)^{-1} \frac{\hat{\mathbf{V}}_b^{niso\prime} \mathbf Z_b^{niso}}{S} \mathbf W_b^{niso} \frac{\mathbf Z_b^{niso\prime} \boldsymbol{\varepsilon}^{niso}_a}{S},
	$$
	where $\dfrac{\mathbf Z_b^{niso\prime} \boldsymbol{\varepsilon}^{niso}_a}{S}$ is asymptotically zero by exogeneity of $\mathbf Z_b^{niso}$. Therefore,
	\begin{equation}\label{eq:plimapsi}
		\plim_{\mathcal{T}_a} \hat{\boldsymbol{\psi}}_b = \plim \left(\hat{\mathbf H}^{-1}_b\, \hat{\mathbf{V}}_b^{niso\prime} \mathbf Z_b^{niso} \mathbf W_b^{niso}\mathbf Z_b^{niso\prime} \mathbf V_a^{niso} \boldsymbol{\psi}_a\right).
	\end{equation}

	As $\boldsymbol\delta_{a,b} = \plim \hat{\boldsymbol{\psi}}_b - \plim_{\mathcal{T}_a} \hat{\boldsymbol{\psi}}_b$, Equations~\eqref{eq:psihat} and~\eqref{eq:plimapsi} imply that:
	\begingroup
	\allowdisplaybreaks
	\begin{align*}
		\boldsymbol\delta_{a,b}& = \plim \left(\hat{\mathbf H}^{-1}_b\, \hat{\mathbf{V}}_b^{niso\prime} \mathbf Z_b^{niso} \mathbf W_b^{niso}\mathbf Z_b^{niso\prime}(\mathbf y^{niso} - \mathbf V_a^{niso} \boldsymbol \psi_a)\right),\\
		\boldsymbol\delta_{a,b} &= \plim \hat{\mathbf H}^{-1}_b\, \hat{\mathbf{V}}_b^{niso\prime} \mathbf Z_b^{niso} \mathbf W_b^{niso}\mathbf Z_b^{niso\prime}\boldsymbol{\varepsilon}^{niso}_a.
	\end{align*}
	\endgroup
	
	\noindent As $\hat{\boldsymbol{\varepsilon}}_a^{niso}$ is a consistent estimator of $\boldsymbol{\varepsilon}^{niso}_a$ component by component, the previous equation suggests that $\boldsymbol{\delta}_{a,b}$ can be estimated by $\hat{\mathbf{H}}_b^{-1} \hat{\mathbf{V}}_b^{niso\prime} \mathbf{Z}_b^{niso} \mathbf{W}_b^{niso} \mathbf{Z}_b^{niso\prime} \hat{\boldsymbol{\varepsilon}}_a^{niso}$. However, componentwise consistency of $\hat{\boldsymbol{\varepsilon}}_a^{niso}$ is not sufficient to justify the substitution of $\boldsymbol{\varepsilon}^{niso}_a$ with $\hat{\boldsymbol{\varepsilon}}_a^{niso}$ in this expression. It is also necessary to ensure that $\dfrac{\mathbf{Z}_b^{niso\prime} \hat{\boldsymbol{\varepsilon}}_a^{niso}}{S}$ consistently estimates $\dfrac{\mathbf{Z}_b^{niso\prime} \boldsymbol{\varepsilon}^{niso}_a)}{S}$. I therefore introduce the following weak conditions.
	
	\begin{assumption}\label{ass:encompassing}The parameters $\boldsymbol{\beta}_1$ and $\boldsymbol{\psi}_a$ lie in compact spaces. Additionally, the variables $\boldsymbol{x}_{s,i}$ and $y_{s,i}$ are uniformly bounded in probability.
	\end{assumption}
	
	\noindent Under Assumption~\ref{ass:encompassing}, the derivatives of $(\mathbf{Z}_b^{niso\prime} \hat{\boldsymbol{\varepsilon}}_a^{niso})/S$ with respect to $\hat{\boldsymbol{\beta}}_1$ and $\hat{\boldsymbol{\psi}}_a$ are bounded in probability, implying that $\plim (\mathbf{Z}_b^{niso\prime} \hat{\boldsymbol{\varepsilon}}_a^{niso})/S = \plim (\mathbf{Z}_b^{niso\prime} \boldsymbol{\varepsilon}_a^{niso})/S$.\footnote{This follows from the mean value theorem, which implies that $(\mathbf{Z}_b^{niso\prime} \hat{\boldsymbol{\varepsilon}}_a^{niso})/S - (\mathbf{Z}_b^{niso\prime} \boldsymbol{\varepsilon}_a^{niso})/S$ is proportional to $\lVert \hat{\boldsymbol{\beta}}_1 - \boldsymbol{\beta}_1 \rVert + \lVert \hat{\boldsymbol{\psi}}_a - \boldsymbol{\psi}_a \rVert$.} As a result, $\hat{\boldsymbol{\delta}}_{a,b} = \hat{\mathbf{H}}_b^{-1} \hat{\mathbf{V}}_b^{niso\prime} \mathbf{Z}_b^{niso} \mathbf{W}_b^{niso} \mathbf{Z}_b^{niso\prime} \hat{\boldsymbol{\varepsilon}}_a^{niso}$ is a consistent estimator of $\boldsymbol{\delta}_{a,b}$. The asymptotic distribution of $\sqrt{S}\, (\hat{\boldsymbol{\delta}}_{a,b} - \boldsymbol{\delta}_{a,b})$ is presented in Online Appendix~\ref{OA:other:encompassing}.

	\bigskip
	\bigskip
	\let\oldthebibliography\thebibliography
	\let\endoldthebibliography\endthebibliography
	\renewenvironment{thebibliography}[1]{
		\begin{oldthebibliography}{#1}
			\setlength{\itemsep}{0.5em}
			\setlength{\parskip}{0em}
		}
		{
		\end{oldthebibliography}
	}
	{\linespread{1}
		\fontsize{11}{14}\selectfont
		\bibliography{References}
		\bibliographystyle{ecta}}
	
	%########################################################
	\newpage
	\setcounter{page}{1}

	\begin{mytitlepage}
		\title{Online Appendix to "\TITLE"}
		\maketitle
		
		\begin{abstract}{\linespread{1.2}\selectfont
				\noindent This supplemental appendix includes additional results and technical details omitted from the main text. Section~\ref{OA:WeightedNet} extends the microfoundations to the case of weighted networks. Section~\ref{OA:other} presents additional econometric results on the limiting distribution of the model estimator, the encompassing test, and the Wald-style test for the validity of the type II instruments.
		}\end{abstract}
	\end{mytitlepage}
	
	\bigskip
	\section{Microeconomic Foundations with Weighted Networks}\label{OA:WeightedNet}
	In this section, I demonstrate the existence and uniqueness of equilibrium and discuss instrumental variables for quantile peer outcomes in the context of weighted networks. The network is represented by an adjacency matrix $\mathbf{G} = [g_{ij}]$ of dimension $n \times n$, where $g_{ij}$ is a nonnegative variable (not necessarily binary) that measures the intensity of the outgoing link from agent $i$ to agent $j$.
	
	\subsection{Sample Weighted Quantiles}\label{OA:QWeightedNet}
	Since the network is weighted, I replace the sample quantile in Equation \eqref{eq:defqtaui} with a sample weighted quantile, where the weight of $y_{j}$ is $g_{ij}$. The definition of the sample quantile in Appendix \ref{Append:sampleQ} %\citepoa{hyndman1996sampleoa} 
	can be generalized to sample weighted quantiles. 
	
	Let $g_{i,(1)}$, \dots, $g_{i,(d_i)}$ be the weights corresponding to $y_{i,(1)}$, \dots, $y_{i,(d_i)}$, respectively, and define $G_{i,l} = \sum_{k = 1}^{l} g_{i,(k)}$. Let $\hat{k}$ be the largest integer such that $\dfrac{G_{i,\hat k}}{G_{i,d_i}}  \leq \tau$. The weighted sample quantile is given by formula \eqref{eq:defqtaui1} but with different weight $\omega_i$ and index $\pi_i$:
	$$q_{\tau,i}(\mathbf y_{-i}) = (1 - \omega_i)y_{i,(\pi_i)}  + \omega_iy_{i,(\pi_i+1)}, \text{ where }$$$$
	\pi_i = \left\lfloor \hat k + 1 - \tau + (\tau G_{i,d_i} - G_{i,\hat k})/g_{i,(\hat k + 1)}\right\rfloor \text{ and } $$$$
	\omega_i = \hat k + 1- \tau + (\tau G_{i,d_i} - G_{i,\hat k})/g_{i,(\hat k + 1)} - \pi_i.$$
	The higher the intensity of the outgoing link from $i$ to $\hat k + 1$, the greater the weight of $y_{i,(\hat k + 1)}$. The weight $\omega_i$ and the index $\pi_i$ depend on $\mathbf y_{-i}$.

	\subsection{Existence and Uniqueness of Equilibrium with Weighted Networks}\label{OA:NEWeightedNet}
	To establish the existence and uniqueness of the Nash equilibrium with weighted networks, I first present the following lemma.
	\begin{lemma}\label{lemmaQuant}
		For any outcomes $\mathbf y = (y_1, ~\dots, ~ y_n)^{\prime}$ and $\tilde{\mathbf y} = (\tilde y_1, ~\dots, ~ \tilde y_{n})^{\prime}$ and a quantile level $\tau \in [0, ~1]$, 
		we have $\lvert q_{\tau,i}(\mathbf y_{-i}) - q_{\tau,i}(\tilde{\mathbf y}_{-i}) \rvert \leq \lVert \mathbf y - \tilde{\mathbf y} \rVert_{\infty}$.
	\end{lemma}
	\begin{proof}
		Following the notation in Section \ref{OA:QWeightedNet}, we have $q_{\tau,i}(\mathbf y_{-i}) = (1 - \omega_i)y_{i,(\pi_i)}  + \omega_iy_{i,(\pi_i+1)}$ and $q_{\tau,i}(\tilde{\mathbf y}_{-i}) = (1 - \tilde\omega_i)\tilde y_{i,(\tilde\pi_i)}  + \tilde\omega_i\tilde y_{i,(\tilde\pi_i+1)}$, where $\tilde\omega_i$ and $\tilde\pi_i$ are the weight and index for $\tilde{\mathbf y}_{-i}$. I first consider the case where $\pi_i \ne \tilde\pi_i$ and then the case where $\pi_i = \tilde\pi_i$.

		\medskip
		\noindent \underline{Case 1}: $\pi_i \ne \tilde\pi_i$. If $y_{i,(\pi_i)}$ and $\tilde y_{i,(\tilde\pi_i)}$ do not have the same rank within $\mathbf y_{-i}$ and $\tilde{\mathbf y}_{-i}$, respectively, then there exist indices $s$ and $p$ such that $y_{s} \leq y_{i,(\pi_i)} \leq y_{i,(\pi_i + 1)} \leq y_{p}$ and $\tilde y_{p} \leq \tilde y_{i,(\tilde\pi_i)} \leq \tilde y_{i,(\tilde\pi_i+1)} \leq \tilde y_{s}$. Thus, $y_{s} \leq q_{\tau,i}(\mathbf y_{-i}) \leq y_{p}$ and $\tilde y_{p} \leq q_{\tau,i}(\tilde{\mathbf y}_{-i}) \leq \tilde y_{s}$. Therefore, $y_{s} - \tilde y_{s} \leq q_{\tau,i}(\mathbf y_{-i}) - q_{\tau,i}(\tilde{\mathbf y}_{-i}) \leq y_{p} - \tilde y_{p}$, which implies that there exists an index $m\in\{s, p\}$ such that $\lvert q_{\tau,i}(\mathbf y_{-i}) - q_{\tau,i}(\tilde{\mathbf y}_{-i}) \rvert \leq \lvert y_{m} - \tilde y_{m} \rvert$. Consequently, $\lvert q_{\tau,i}(\mathbf y_{-i}) - q_{\tau,i}(\tilde{\mathbf y}_{-i}) \rvert \leq \lVert \mathbf y - \tilde{\mathbf y} \rVert_{\infty}$.
		
		\medskip
		\noindent \underline{Case 2}: $\pi_i = \tilde\pi_i$. If, in addition, $\omega_i = \tilde \omega_i$, then $\lvert q_{\tau,i}(\mathbf y_{-i}) - q_{\tau,i}(\tilde{\mathbf y}_{-i}) \rvert \leq (1 - \omega_i)\lvert y_{i,(\pi_i)} - \tilde y_{i,(\pi_i)}\rvert + \omega_i\lvert y_{i,(\pi_i + 1)} - \tilde y_{i,(\pi_i + 1)}\rvert $. By Lemma \ref{lemmaRank}, there exist indices $s$ and $p$, such that  $\lvert y_{i,(\pi_i)} - \tilde y_{i,(\pi_i)}\rvert \leq \lvert y_{s} - \tilde y_s \rvert$ and $\lvert y_{i,(\pi_i + 1)} - \tilde y_{i,(\pi_i + 1)}\rvert \leq \lvert y_{p} - \tilde y_p \rvert$. Therefore, $\lvert q_{\tau,i}(\mathbf y_{-i}) - q_{\tau,i}(\tilde{\mathbf y}_{-i}) \rvert \leq \lVert \mathbf y - \tilde{\mathbf y} \rVert_{\infty}$.
		
		Assume now that $\omega_i \ne \tilde \omega_i$. The only way this can hold while $\pi_i = \tilde\pi_i$ is if $\mathbf y_{-i}$ and $\tilde{\mathbf y}_{-i}$ are sorted differently, such that there exist indices $s$ and $p$ with $y_{s} \leq y_{i,(\pi_i)} \leq y_{i,(\pi_i + 1)} \leq y_{p}$ and $\tilde y_{p} \leq \tilde y_{i,(\tilde\pi_i)} \leq \tilde y_{i,(\tilde\pi_i+1)} \leq \tilde y_{s}$. Indeed, the condition $\omega_i \ne \tilde \omega_i$ requires that the sum of the weights of the $\pi_i$ (or, respectively, $\pi_i + 1$) smallest elements of $\mathbf y_{-i}$ be different from the sum of the weights of the $\pi_i$ (or, respectively, $\pi_i + 1$) smallest elements of $\tilde{\mathbf y}_{-i}$. As in Case 1, if $y_{s} \leq y_{i,(\pi_i)} \leq y_{i,(\pi_i + 1)} \leq y_{p}$ and $\tilde y_{p} \leq \tilde y_{i,(\tilde\pi_i)} \leq \tilde y_{i,(\tilde\pi_i+1)} \leq \tilde y_{s}$, then $\lvert q_{\tau,i}(\mathbf y_{-i}) - q_{\tau,i}(\tilde{\mathbf y}_{-i}) \rvert \leq \lVert \mathbf y - \tilde{\mathbf y} \rVert_{\infty}$. This completes the proof of the lemma.\end{proof}
	
	Using  Lemma \ref{lemmaQuant}, I establish the existence and uniqueness of equilibrium.
	
	\begin{proposition}\label{prop:NEWeight} Assume that the network is weighted and that the utility function of each individual $i\in\mathcal{N}$ is given \eqref{eq:utility}, with $\sum_{\tau \in \mathcal{T}}\lvert \lambda_{\tau} \rvert < 1$. Then, there exists a unique Nash equilibrium $\mathbf y^{\ast} = (y_1^{\ast}, ~\dots,   y_n^{\ast})^{\prime}$ such that $y_i^{\ast} = BR_i(\mathbf{y}^{\ast}_{-i})$ for all $i$.
	\end{proposition}
	
	\begin{proof}
		Let $BR$ be the mapping from $\mathbb R^n$ into itself, defined by $BR(\mathbf y) = (BR_1(\mathbf y_{-1}), ~\dots, \break BR_n(\mathbf y_{-n}))^{\prime}$. As in Appendix \ref{Append:prop:NE}, it is sufficient to demonstrate that $BR$ has a unique fixed point.   We have 
		\begingroup
		\allowdisplaybreaks
		\begin{align*}
			\lVert BR(\mathbf y) - BR(\tilde{\mathbf y}) \rVert_{\infty} &= \max_i \lvert BR_i(\mathbf y_{-i}) - BR_i(\tilde{\mathbf y}_{-i}) \rvert,\nonumber \\
			\lVert BR(\mathbf y) - BR(\tilde{\mathbf y}) \rVert_{\infty} & \textstyle \leq \sum_{\tau \in \mathcal{T}}\left\{\lvert \lambda_{\tau,1} + \lambda_{\tau,2}\rvert  \max_i  \lvert q_{i,\tau}(\mathbf y_{-i}) - q_{i,\tau}(\tilde{\mathbf y}_{-i}) \rvert\right\} 
		\end{align*}  
		\endgroup
		By Lemma \ref{lemmaQuant}, $\lvert q_{\tau,i}(\mathbf y_{-i}) - q_{\tau,i}(\tilde{\mathbf y}_{-i}) \rvert \leq \lVert \mathbf y - \tilde{\mathbf y} \rVert_{\infty}$, which implies that \break $\lVert BR(\mathbf y) - BR(\tilde{\mathbf y}) \rVert_{\infty}  \textstyle \leq  \sum_{\tau \in \mathcal{T}}\lvert \lambda_{\tau,1} + \lambda_{\tau,2}\rvert  \lVert \mathbf y -  \tilde{\mathbf y}\rVert_{\infty}$. Since $\sum_{\tau \in \mathcal{T}}\lvert \lambda_{\tau,1} + \lambda_{\tau,2}\rvert < 1$, it follows that $BR$ is a contraction. By the contraction mapping theorem, it has a unique fixed point. As a result, the game has a unique Nash equilibrium.
	\end{proof}
	
	\section{Supplementary Results on the Econometric Model} \label{OA:other}
	For notational ease, I assume throughout this section that Equations~\eqref{eq:model:iso} and~\eqref{eq:model:niso} are expressed in deviations from the group mean, so that the fixed effect terms are eliminated. I use the superscript $iso$ to denote variables corresponding to isolated agents, and the superscript $niso$ for variables corresponding to non-isolated agents.
	
	\subsection{Limiting Distribution of the Parameter Estimator}
	\label{OA:other:limit}
	
	Estimation proceeds in two stages. The first stage is an ordinary least-squares (OLS) regression based on the equation:
	$$
	\mathbf{y}^{iso} = \mathbf{X}^{iso} \boldsymbol{\beta}_1 + \boldsymbol{\varepsilon}^{iso},
	$$
	where $\mathbf{y}^{iso}$ is the vector of $y_{s,i}^{iso}$, $\mathbf{X}^{iso}$ is the matrix whose rows are $\boldsymbol{x}_{s,i}^{iso}$, and $\boldsymbol{\varepsilon}^{iso}$ is the vector of $\varepsilon_{s,i}^{iso}$. Let $\hat{\boldsymbol{\beta}}_1$ denote the OLS estimator of $\boldsymbol{\beta}_1$.
	
	The second stage is a generalized method of moments (GMM) estimation based on the equation:
	$$
	\mathbf{y}^{niso} = \mathbf{V}^{niso} \boldsymbol{\psi} + \boldsymbol{\varepsilon}^{niso},
	$$
	where $\mathbf{y}^{niso}$ is the vector of $y_{s,i}^{niso}$, and $\mathbf{V}^{niso}$ is the matrix whose row corresponding to agent $i$ in subnetwork $s$ is $(q_{s,i,\tau_1}^{niso},~\dots,~q_{s,i,\tau_{d_t}}^{niso},~\boldsymbol{x}_{s,i}^{niso\prime} \boldsymbol{\beta}_1,~\bar{\boldsymbol{x}}_{s,i}^{niso\prime})$. The parameter vector is $\boldsymbol{\psi} = (\lambda_{\tau_1},~\dots,~\lambda_{\tau_{d_t}},~1 - \lambda_2,~\tilde{\boldsymbol{\beta}}_2)^{\prime}$, and $\boldsymbol{\varepsilon}^{niso}$ is the vector of $\varepsilon_{s,i}^{niso}$. 
	
	The matrix $\mathbf{V}^{niso}$ depends on the unknown parameter $\boldsymbol{\beta}_1$. The GMM estimation is carried out by replacing $\boldsymbol{\beta}_1$ with its OLS estimator $\hat{\boldsymbol{\beta}}_1$. Let $\hat{\mathbf{V}}^{niso}$ denote the matrix obtained by substituting $\hat{\boldsymbol{\beta}}_1$ for $\boldsymbol{\beta}_1$ in $\mathbf{V}^{niso}$, and let $\hat{\boldsymbol{\psi}}$ denote the resulting GMM estimator of $\boldsymbol{\psi}$.
	
	The first-order conditions from the OLS and GMM steps imply:
	\begingroup
	\allowdisplaybreaks
	\begin{align}\label{eq:foc}
		\begin{split}
			&\mathbf{X}^{iso\prime} \mathbf{X}^{iso} (\hat{\boldsymbol{\beta}}_1 - \boldsymbol{\beta}_1) = \mathbf{X}^{iso\prime} \boldsymbol{\varepsilon}^{iso}, \\
			&\hat{\mathbf{V}}^{niso\prime} \mathbf{Z}^{niso} \mathbf{W}^{niso} \mathbf{Z}^{niso\prime} (\hat{\mathbf{V}}^{niso} \hat{\boldsymbol{\psi}} - \mathbf{V}^{niso} \boldsymbol{\psi}) = \hat{\mathbf{V}}^{niso\prime} \mathbf{Z}^{niso} \mathbf{W}^{niso} \mathbf{Z}^{niso\prime} \boldsymbol{\varepsilon}^{niso},
		\end{split}
	\end{align}
	\endgroup
	where $\mathbf{Z}^{niso}$ is the instrument matrix and $\mathbf{W}^{niso}$ is the GMM weighting matrix.

	Let $\boldsymbol{\Gamma} = (\boldsymbol{\beta}_1^{\prime}, ~ \boldsymbol{\psi}^{\prime})^{\prime}$ be the full parameter vector and $\hat{\boldsymbol{\Gamma}} = (\hat{\boldsymbol{\beta}}_1^{\prime}, ~ \hat{\boldsymbol{\psi}}^{\prime})^{\prime}$ its estimator. By applying the mean value theorem to $\hat{\mathbf{V}}^{niso} \hat{\boldsymbol{\psi}} - \mathbf{V}^{niso} \boldsymbol{\psi}$, with respect to $\hat{\boldsymbol{\Gamma}}$, \eqref{eq:foc} implies:
	$$
	\underbrace{\begin{pmatrix}
			\mathbf{I}_{d_1} & \mathbf{0} \\
			\mathbf{0} & \hat{\mathbf{V}}^{niso\prime} \mathbf{Z}^{niso} \mathbf{W}^{niso}
	\end{pmatrix}}_{\mathbb{B}}
	\underbrace{\begin{pmatrix}
			\mathbf{X}^{iso\prime} \mathbf{X}^{iso} & \mathbf{0} \\
			(1 - \lambda_2) \mathbf{Z}^{niso\prime} \mathbf{X}^{niso} & \mathbf{Z}^{niso\prime} \mathbf{V}^{niso}
	\end{pmatrix}}_{\mathbb{F}}
	(\hat{\boldsymbol{\Gamma}} - \boldsymbol{\Gamma}) = \mathbb{B} \boldsymbol{u},
	$$
	where $\mathbf{I}_{d_1}$ is the $d_1 \times d_1$ identity matrix, and
	$$
	\boldsymbol{u} = \begin{pmatrix}
		\mathbf{X}^{iso\prime} \boldsymbol{\varepsilon}^{iso} \\
		\mathbf{Z}^{niso\prime} \boldsymbol{\varepsilon}^{niso}
	\end{pmatrix}.
	$$

	By applying a standard central limit theorem to $\boldsymbol{u} / \sqrt{S}$, it follows that $\sqrt{S}(\hat{\boldsymbol{\Gamma}} - \boldsymbol{\Gamma})$ is asymptotically normally distributed with mean zero and asymptotic variance
	$$
	\left( \frac{\mathbb{B}}{S}\, \frac{\mathbb{F}}{S} \right)^{-1}
	\frac{\mathbb{B}}{S} \, \mathbb{V}ar(\boldsymbol{u}/\sqrt{S})
	\frac{\mathbb{B}^\prime}{S} \,
	\left( \frac{\mathbb{F}^\prime}{S} \,\frac{\mathbb{B}^\prime}{S} \right)^{-1},
	$$
	where $\mathbb{V}ar \left( \boldsymbol{u} / \sqrt{S} \right)$ denotes the asymptotic variance of $\boldsymbol{u} / \sqrt{S}$, which can be easily estimated as is the case in standard GMM procedures. 
	
	The asymptotic variance of the structural parameters can then be obtained from the asymptotic variance of $\sqrt{S}(\hat{\boldsymbol{\Gamma}} - \boldsymbol{\Gamma})$ using the Delta method.
	
	\subsection{\texorpdfstring{Asymptotic Distribution of $\sqrt{S}\, (\hat{\boldsymbol{\delta}}_{a,b} - \boldsymbol{\delta}_{a,b})$}{Asymptotic Distribution of sqrt(S)(delta_ab hat - delta_ab)}}
	\label{OA:other:encompassing}
	The estimator $\hat{\boldsymbol{\delta}}_{a,b}$ is given by
	\begin{equation}\label{eq:delta}
		\hat{\boldsymbol{\delta}}_{a,b} = \hat{\mathbf{H}}_b^{-1} \hat{\mathbf{V}}_b^{niso\prime} \mathbf{Z}_b^{niso} \mathbf{W}_b^{niso} \mathbf{Z}_b^{niso\prime} \hat{\boldsymbol{\varepsilon}}_a^{niso},
	\end{equation}
	where $\hat{\mathbf{H}}_b = \hat{\mathbf{V}}_b^{niso\prime} \mathbf{Z}_b^{niso} \mathbf{W}_b^{niso} \mathbf{Z}_b^{niso\prime} \hat{\mathbf{V}}_b^{niso}$. Given that $\hat{\boldsymbol{\varepsilon}}_a^{niso} = \mathbf y^{niso} - \hat{\mathbf V}_a^{niso} \hat{\boldsymbol{\psi}}_a$ and the representation of $\mathbf y^{niso}$ under $\mathcal{T}_a$ is $\mathbf y^{niso} = \mathbf V_a^{niso} \boldsymbol{\psi}_a + \boldsymbol{\varepsilon}_a^{niso}$, it follows that $\hat{\boldsymbol{\varepsilon}}_a^{niso} = \mathbf V_a^{niso} \boldsymbol{\psi}_a - \hat{\mathbf V}_a^{niso} \hat{\boldsymbol{\psi}}_a + \boldsymbol{\varepsilon}_a^{niso}$. Substituting $\hat{\boldsymbol{\varepsilon}}_a^{niso}$ into \eqref{eq:delta} and rearranging terms yields
	\begin{align}
		&\hat{\mathbf{H}}_b\hat{\boldsymbol{\delta}}_{a,b} = \hat{\mathbf{V}}_b^{niso\prime} \mathbf{Z}_b^{niso} \mathbf{W}_b^{niso} \mathbf{Z}_b^{niso\prime} (\mathbf V_a^{niso} \boldsymbol{\psi}_a - \hat{\mathbf V}_a^{niso} \hat{\boldsymbol{\psi}}_a + \boldsymbol{\varepsilon}_a^{niso}),\nonumber\\
		& \hat{\mathbf{V}}_b^{niso\prime} \mathbf{Z}_b^{niso} \mathbf{W}_b^{niso} \mathbf{Z}_b^{niso\prime}(\hat{\mathbf V}_b^{niso}\hat{\boldsymbol{\delta}}_{a,b}  + \hat{\mathbf V}_a^{niso} \hat{\boldsymbol{\psi}}_a - \mathbf V_a^{niso} \boldsymbol{\psi}_a) = \hat{\mathbf{V}}_b^{niso\prime} \mathbf{Z}_b^{niso} \mathbf{W}_b^{niso} \mathbf{Z}_b^{niso\prime} \boldsymbol{\varepsilon}_a^{niso},\nonumber\\
		\begin{split}& \mathbb B^{\ast}_{b} \mathbf{Z}_b^{niso\prime}(\hat{\mathbf V}_b^{niso}\hat{\boldsymbol{\delta}}_{a,b} - \mathbf V_b^{niso}\boldsymbol{\delta}_{a,b}  + \hat{\mathbf V}_a^{niso} \hat{\boldsymbol{\psi}}_a - \mathbf V_a^{niso} \boldsymbol{\psi}_a) = \\ & 
			\quad\quad\quad\quad \mathbb B^{\ast}_{b} \mathbf{Z}_b^{niso\prime} (\boldsymbol{\varepsilon}_a^{niso} - \mathbf V_b^{niso}\boldsymbol{\delta}_{a,b}),
		\end{split} \label{eq:ortho:delta}
	\end{align}
	where $\mathbb{B}^{\ast}_{b} = \hat{\mathbf{V}}_b^{niso\prime} \mathbf{Z}_b^{niso} \mathbf{W}_b^{niso}$. Equation~\eqref{eq:ortho:delta} is the orthogonality condition for $\hat{\boldsymbol{\delta}}_{a,b}$. This equation depends on the estimators $\hat{\boldsymbol{\beta}}_1$ and $\hat{\boldsymbol{\psi}}_a$. The orthogonality conditions for these estimators follow from Equation~\eqref{eq:foc}:
	\begin{align}\label{eq:ortho:psia}
		\begin{split}
			&\mathbf{X}^{iso\prime} \mathbf{X}^{iso} (\hat{\boldsymbol{\beta}}_1 - \boldsymbol{\beta}_1) = \mathbf{X}^{iso\prime} \boldsymbol{\varepsilon}^{iso}, \\
			&\mathbb B^{\ast}_{a} \mathbf{Z}^{niso\prime}_a (\hat{\mathbf{V}}^{niso}_a \hat{\boldsymbol{\psi}}_a - \mathbf{V}^{niso}_a \boldsymbol{\psi}_a) = \mathbb B^{\ast}_{a} \mathbf{Z}^{niso\prime}_a \boldsymbol{\varepsilon}^{niso}_a,
		\end{split}
	\end{align}
	where $\mathbb B^{\ast}_{a} = \hat{\mathbf{V}}_a^{niso\prime} \mathbf{Z}_a^{niso} \mathbf{W}_a^{niso}$. Let $\boldsymbol{\Gamma}^{\ast} = (\boldsymbol{\beta}_1^{\prime}, ~ \boldsymbol{\psi}^{\prime}_a, ~\boldsymbol{\delta}_{a,b})^{\prime}$ be the full parameter vector, and let $\hat{\boldsymbol{\Gamma}}^{\ast} = (\hat{\boldsymbol{\beta}}_1^{\prime}, ~ \hat{\boldsymbol{\psi}}_a^{\prime}, ~\hat{\boldsymbol{\delta}}_{a,b})^{\prime}$ be its estimator. By applying the mean value theorem to $\hat{\mathbf V}_b^{niso}\hat{\boldsymbol{\delta}}_{a,b} - \mathbf V_b^{niso}\boldsymbol{\delta}_{a,b}$ and $\hat{\mathbf V}_a^{niso} \hat{\boldsymbol{\psi}}_a - \mathbf V_a^{niso} \boldsymbol{\psi}_a$ with respect to $\hat{\boldsymbol{\Gamma}}^{\ast}$ \eqref{eq:ortho:delta} and \eqref{eq:ortho:psia} imply:
	$$
	\underbrace{\begin{pmatrix}
			\mathbf{I}_{d_1} & \mathbf{0} & \mathbf{0} \\
			\mathbf{0} & \mathbb B^{\ast}_{a} & \mathbf{0} \\
			\mathbf{0} & \mathbf{0} & \mathbb B^{\ast}_{b}
	\end{pmatrix}}_{\mathbb{B}^{\ast}}
	\underbrace{\begin{pmatrix}
			\mathbf{X}^{iso\prime} \mathbf{X}^{iso} & \mathbf{0} & \mathbf{0} \\
			(1 - \lambda_{2,a}) \mathbb{F}^{\ast}_{a1} & \mathbb{F}^{\ast}_{a2} & \mathbf{0} \\
			(1 - \lambda_{2,a} + \delta_{a,b,1})\mathbb{F}^{\ast}_{b1} & \mathbb{F}^{\ast}_{b2} & \mathbb{F}^{\ast}_{b3}\end{pmatrix}}_{\mathbb{F}^{\ast}}
	(\hat{\boldsymbol{\Gamma}}^{\ast} - \boldsymbol{\Gamma}^{\ast}) = \mathbb{B}^{\ast} \boldsymbol{u}^{\ast},
	$$
	where $\mathbb{F}^{\ast}_{a1} = \mathbf{Z}^{niso\prime}_a \mathbf{X}^{niso}$, $\mathbb{F}^{\ast}_{a2} = \mathbf{Z}^{niso\prime}_a \mathbf{V}^{niso}_a$, $\mathbb{F}^{\ast}_{b1} = \mathbf{Z}^{niso\prime}_b \mathbf{X}^{niso}$, $\mathbb{F}^{\ast}_{b2} = \mathbf{Z}^{niso\prime}_b \mathbf{V}^{niso}_a$ and $\mathbb{F}^{\ast}_{b3} = \mathbf{Z}^{niso\prime}_b \mathbf{V}^{niso}_b$. The terms $1 - \lambda_{2,a}$ and $\delta_{a,b,1}$ are the coefficients associated with $\boldsymbol{x}_{s,i}^{\prime} \boldsymbol{\beta}_1$ in $\boldsymbol{\psi}_a$ and $\boldsymbol{\delta}_{a,b}$, respectively. The vector $\boldsymbol{u}^{\ast}$ is given by
	$$
	\boldsymbol{u}^{\ast} = \begin{pmatrix}
		\mathbf{X}^{iso\prime} \boldsymbol{\varepsilon}^{iso} \\
		\mathbf{Z}^{niso\prime}_a \boldsymbol{\varepsilon}^{niso}_a \\
		\mathbf{Z}^{niso\prime}_b (\boldsymbol{\varepsilon}_a^{niso} - \mathbf V_b^{niso}\boldsymbol{\delta}_{a,b})
	\end{pmatrix}.
	$$
	
	By applying a standard central limit theorem to $\boldsymbol{u}^{\ast} / \sqrt{S}$, it follows that the asymptotic variance of $\sqrt{S} (\hat{\boldsymbol{\Gamma}}^{\ast} - \boldsymbol{\Gamma}^{\ast})$ is
	$$
	\left( \frac{\mathbb{B^{\ast}}}{S}\, \frac{\mathbb{F}^{\ast}}{S} \right)^{-1}
	\frac{\mathbb{B}^{\ast}}{S} \, \mathbb{V}ar(\boldsymbol{u}^{\ast}/\sqrt{S})
	\frac{\mathbb{B}^{\ast\prime}}{S} \,
	\left( \frac{\mathbb{F}^{\ast\prime}}{S} \,\frac{\mathbb{B}^{\ast\prime}}{S} \right)^{-1},
	$$
	where $\mathbb{V}ar(\boldsymbol{u}^{\ast} / \sqrt{S})$ denotes the asymptotic variance of $\boldsymbol{u}^{\ast} / \sqrt{S}$. From this result, the asymptotic variance of $\sqrt{S} (\hat{\boldsymbol{\delta}}_{a,b} - \boldsymbol{\delta}_{a,b})$ can be directly obtained, as it is to a subvector of $\sqrt{S} (\hat{\boldsymbol{\Gamma}}^{\ast} - \boldsymbol{\Gamma}^{\ast})$.

	\subsection{Variance of \texorpdfstring{$\sqrt{S}(\hat{\boldsymbol{\psi}}_1 - \hat{\boldsymbol{\psi}}_2)$}{the Parameter Difference}}
	\label{OA:other:variance}
	As $\mathbf{Z}_1$ is valid, $\hat{\boldsymbol{\psi}}_1$ converges in probability to $\boldsymbol{\psi}$. Let $\boldsymbol{\psi}^{\dagger}$ denote the probability limit of $\hat{\boldsymbol{\psi}}_2$, which may differ from $\boldsymbol{\psi}$ if $\mathbf{Z}_1$ is endogenous. Let $\boldsymbol{\Gamma}^{\dagger} = (\boldsymbol{\beta}_1^{\prime}, ~ \boldsymbol{\psi}^{\prime}, ~\boldsymbol{\psi}^{\dagger\prime})^{\prime}$ denote the full parameter vector, and let $\hat{\boldsymbol{\Gamma}}^{\dagger} = (\hat{\boldsymbol{\beta}}_1^{\prime}, ~ \hat{\boldsymbol{\psi}}_1^{\prime}, ~\hat{\boldsymbol{\psi}}_2^{\prime})^{\prime}$ be its consistent estimator.

	The orthogonality conditions for $\hat{\boldsymbol{\beta}}_1$, $\hat{\boldsymbol{\psi}}_1$, and $\hat{\boldsymbol{\psi}}_2$ follow from~\eqref{eq:foc}. Applying the mean value theorem to the resulting equations with respect to $\hat{\boldsymbol{\Gamma}}^{\dagger}$ yields:
	$$
	\underbrace{\begin{pmatrix}
			\mathbf{I}_{d_1} & \mathbf{0} & \mathbf{0} \\
			\mathbf{0} & \mathbb B^{\dagger}_{2} & \mathbf{0} \\
			\mathbf{0} & \mathbf{0} & \mathbb B^{\dagger}_{3}
	\end{pmatrix}}_{\mathbb{B}^{\dagger}}
	\underbrace{\begin{pmatrix}
			\mathbf{X}^{iso\prime} \mathbf{X}^{iso} & \mathbf{0} & \mathbf{0} \\
			(1 - \lambda_{2,1}) \mathbb{F}^{\dagger}_{21} & \mathbb{F}^{\dagger}_{22} & \mathbf{0} \\
			(1 - \lambda_{2,2})\mathbb{F}^{\dagger}_{31} & \mathbf{0} & \mathbb{F}^{\dagger}_{33}\end{pmatrix}}_{\mathbb{F}^{\dagger}}
	(\hat{\boldsymbol{\Gamma}}^{\dagger} - \boldsymbol{\Gamma}^{\dagger}) = \mathbb{B}^{\dagger} \boldsymbol{u}^{\dagger},
	$$
	where $\mathbb{B}^{\dagger}_{k} = \hat{\mathbf{V}}^{niso\prime} \mathbf{Z}^{niso}_k \mathbf{W}^{niso}_k$, $\mathbb{F}^{\dagger}_{k1} = \mathbf{Z}^{niso\prime}_k \mathbf{X}^{niso}$, and $\mathbb{F}^{\dagger}_{kk} = \mathbf{Z}^{niso\prime}_k \mathbf{V}^{niso}$ for all $k \in \{1, 2\}$. The matrix $\mathbf{W}^{niso}_k$ is the GMM weighting matrix corresponding to $\mathbf{Z}^{niso}_k$, and $1 - \lambda_{2,k}$ is the coefficient associated with $\boldsymbol{x}_{s,i}^{\prime} \boldsymbol{\beta}_1$ in $\boldsymbol{\psi}_k$. The vector $\boldsymbol{u}^{\dagger}$ is given by
	$$\boldsymbol{u}^{\dagger} = \left(\mathbf{X}^{iso\prime} \boldsymbol{\varepsilon}^{iso},  ~\mathbf{Z}^{niso\prime}_1 \boldsymbol{\varepsilon}^{niso}, ~\mathbf{Z}^{niso\prime}_2 \boldsymbol{\varepsilon}^{niso}\right)\prime$$
	
	Applying a standard central limit theorem to $\boldsymbol{u}^{\dagger} / \sqrt{S}$, it follows that the asymptotic variance of $\sqrt{S} (\hat{\boldsymbol{\Gamma}}^{\dagger} - \boldsymbol{\Gamma}^{\dagger})$ is
	$$
	\left( \frac{\mathbb{B^{\dagger}}}{S}\, \frac{\mathbb{F}^{\dagger}}{S} \right)^{-1}
	\frac{\mathbb{B}^{\dagger}}{S} \, \mathbb{V}ar(\boldsymbol{u}^{\dagger}/\sqrt{S})
	\frac{\mathbb{B}^{\dagger\prime}}{S} \,
	\left( \frac{\mathbb{F}^{\dagger\prime}}{S} \,\frac{\mathbb{B}^{\dagger\prime}}{S} \right)^{-1},
	$$
	where $\mathbb{V}ar(\boldsymbol{u}^{\dagger} / \sqrt{S})$ denotes the asymptotic variance of $\boldsymbol{u}^{\dagger} / \sqrt{S}$. From this expression, the variance of $\sqrt{S} (\hat{\boldsymbol{\psi}}_1 - \hat{\boldsymbol{\psi}}_2)$ can be consistently estimated using the Delta method.

	%\bigskip
	%\bigskip
	%{\linespread{1}
		%\fontsize{11}{14}\selectfont
		%\bibliographyoa{Referencesoa}
		%\bibliographystyleoa{ecta}}
\end{document}